\documentclass[11pt]{amsart}
\usepackage{float}
\usepackage[T1]{fontenc}

\usepackage{fullpage}
\usepackage{amsmath,amsthm,amssymb,refcount}
\usepackage{xspace}
\usepackage{hyperref}
\usepackage{microtype}
\usepackage[utf8]{inputenc}
\usepackage[T1]{fontenc}
\usepackage{pdfsync}
\usepackage{enumitem}
\usepackage{stmaryrd}
\usepackage{xifthen}
\usepackage{bbm}
\usepackage{amsmath}
\usepackage{xfrac}
\usepackage[mathscr]{euscript}
\usepackage{tikz}
\usetikzlibrary{decorations.pathreplacing,hobby,backgrounds,trees,arrows,automata,shapes,positioning,fit,patterns,calc,matrix}

\usepackage{xcolor}
\definecolor{bookgreen}{RGB}{40,130,80}
\definecolor{bookblue}{RGB}{50,110,150}
\definecolor{bookred}{RGB}{180,15,47}

\newcounter{sauvegarde}
\newcommand\adjustc[1]{
  \setcounter{sauvegarde}{\thetheorem}
  \setcounterref{theorem}{#1}
  \addtocounter{theorem}{-1}
}

\newcommand\restorec{
  \setcounter{theorem}{\thesauvegarde}
}

\newcommand{\upset}[2][\Cs]{\ensuremath{{\mathord\uparrow_{#1}\,#2}}}

\newcommand{\canod}{\ensuremath{\leqslant_\Ds}\xspace}
\newcommand{\canof}{\ensuremath{\leqslant_\Fs}\xspace}

\newcommand{\cocl}[1]{\ensuremath{\mathit{co\textup{-}}\!#1}\xspace}

\newcommand{\fo}{\ensuremath{\textup{FO}}\xspace}

\newcommand{\MOD}{\textup{MOD}\xspace}

\newcommand{\fowm}{\mbox{\ensuremath{\fo({<},\MOD)}}\xspace}

\newcommand{\siwm}[1]{\ensuremath{\Sigma_{#1}({<},\MOD)}\xspace}

\newcommand{\siwmu}{\siwm{1}}

\newcommand{\siwmd}{\siwm{2}}

\newcommand{\siwmt}{\siwm{3}}

\newcommand{\bswm}[1]{\ensuremath{\Bs\Sigma_{#1}({<},\MOD)}\xspace}

\newcommand{\bswmu}{\bswm{1}}

\newcommand{\md}{\ensuremath{\textup{\MOD}}\xspace}

\newcommand{\grp}{\ensuremath{\textup{GR}}\xspace}

\newcommand{\sfr}{\ensuremath{\textup{SF}}\xspace}

\newcommand{\bool}[1]{\ensuremath{\mathit{Bool}(#1)}\xspace}
\newcommand{\pol}[1]{\ensuremath{\mathit{Pol}(#1)}\xspace}
\newcommand{\bpol}[1]{\ensuremath{\mathit{BPol}(#1)}\xspace}
\newcommand{\pbpol}[1]{\ensuremath{\mathit{PBPol}(#1)}\xspace}

\newcommand{\copol}[1]{\ensuremath{\mathit{co\textup{-}\!Pol}(#1)}\xspace}

\newcommand{\imprint}{imprint\xspace}
\newcommand{\imprints}{imprints\xspace}

\newcommand{\tame}{multiplicative\xspace}

\newcommand{\Ratms}{Rating maps\xspace}

\newcommand{\ratms}{rating maps\xspace}
\newcommand{\ratm}{rating map\xspace}

\newcommand{\Nice}{Nice\xspace}
\newcommand{\nice}{nice\xspace}

\newcommand{\mratm}{multiplicative rating map\xspace}
\newcommand{\mratms}{multiplicative rating maps\xspace}

\newcommand{\Mratms}{Multiplicative rating maps\xspace}

\newcommand{\iden}{\veps-approximation\xspace}
\newcommand{\idens}{\veps-approximations\xspace}

\newcommand{\prin}[2]{\ensuremath{\Is[#1](#2)}\xspace}

\newcommand{\opti}[2]{\ensuremath{\Is_{#1}\left[#2\right]}\xspace}
\newcommand{\copti}[1]{\opti{\Cs}{#1}}

\newcommand{\popti}[3]{\ensuremath{\Ps_{#1}^{#2}[#3]}\xspace}

\newcommand{\iopti}[2]{\ensuremath{\fri_{#1}[#2]}\xspace}
\newcommand{\ioptic}[1]{\iopti{\Cs}{#1}}

\newcommand{\pocopti}{\popti{\pol{\Cs}}{\alpha}{\rho}}

\newcommand{\bpoluopti}{\ioptic{\bratauxbc}}
\newcommand{\pbpoluopti}{\ioptic{\lratauxppc}}

\newcommand{\typ}[2]{\ensuremath{[#1]_{#2}}\xspace}

\usepackage{stmaryrd}

\newcommand{\Bs}{\ensuremath{\mathcal{B}}\xspace}
\newcommand{\Cs}{\ensuremath{\mathcal{C}}\xspace}
\newcommand{\Ds}{\ensuremath{\mathcal{D}}\xspace}

\newcommand{\Fs}{\ensuremath{\mathcal{F}}\xspace}
\newcommand{\Gs}{\ensuremath{\mathcal{G}}\xspace}

\newcommand{\Is}{\ensuremath{\mathcal{I}}\xspace}

\newcommand{\Ps}{\ensuremath{\mathcal{P}}\xspace}

\newcommand{\Ts}{\ensuremath{\mathcal{T}}\xspace}

\newcommand{\Hb}{\ensuremath{\mathbf{H}}\xspace}

\newcommand{\Kb}{\ensuremath{\mathbf{K}}\xspace}
\newcommand{\Lb}{\ensuremath{\mathbf{L}}\xspace}

\newcommand{\fri}{\ensuremath{\mathbbm{i}}\xspace}

\newcommand{\vari}{quotient-closed Boolean algebra\xspace}

\newcommand{\varis}{quotient-closed Boolean algebras\xspace}
\newcommand{\Varis}{Quotient-closed Boolean algebras\xspace}

\newcommand{\pvari}{quotient-closed lattice\xspace}

\newcommand{\pvaris}{quotient-closed lattices\xspace}

\newcommand{\nat}{\ensuremath{\mathbb{N}}\xspace}
\def\inv{^{-1}}

\newcommand{\veps}{\ensuremath{\varepsilon}\xspace}

\newcommand{\dclosp}[1]{\ensuremath{\mathord{\downarrow_{#1}}}\xspace}
\newcommand{\dclosr}{\dclosp{R}}

\newcommand{\brataux}[2]{\ensuremath{\xi_{#1}[#2]}\xspace}

\newcommand{\bratauxd}{\brataux{\Ds}{\rho}}
\newcommand{\bratauxbc}{\brataux{\bpol{\Cs}}{\rho}}

\newcommand{\lrataux}[3]{\ensuremath{\zeta_{#1}^{#2}[#3]}\xspace}

\newcommand{\lratauxd}{\lrataux{\Ds}{\alpha}{\rho}}

\newcommand{\lratauxppc}{\lrataux{\pbpol{\Cs}}{\alpha}{\rho}}

\newcommand{\quasi}[1]{\ensuremath{\mu_{#1}}\xspace}
\newcommand{\quasir}{\quasi{\rho}}

\theoremstyle{plain}
\newtheorem{theorem}{Theorem}
\newtheorem{corollary}[theorem]{Corollary}
\newtheorem{fact}[theorem]{Fact}
\newtheorem{proposition}[theorem]{Proposition}
\newtheorem{lemma}[theorem]{Lemma}
\newtheorem{remark}[theorem]{Remark}
\newtheorem{example}[theorem]{Example}

\begin{document}

\title{Separation and covering for group based concatenation hierarchies}
\author{Thomas Place}
\address{LaBRI, Bordeaux University and IUF}
\author{Marc Zeitoun}
\address{LaBRI, Bordeaux University}
\thanks{Both authors acknowledge support from the DeLTA project (ANR-16-CE40-0007).}

\maketitle

\begin{abstract}
  Concatenation hierarchies are natural classifications of regular languages. All such hierarchies are built through the same construction process: one starts from an initial, specific class of languages (the basis) and builds new levels using two generic operations. Concatenation hierarchies have gathered a lot of interest since the early 70s, notably thanks to an alternate logical definition: each concatenation hierarchy can be defined as the quantification alternation hierarchy within a variant of first-order logic over words (while the hierarchies differ by their bases, the variants differ by their set of available~predicates).

  Our goal is to understand these hierarchies. A typical approach is to look at two decision problems: membership and separation. In the paper we are interested in the latter, which is more general. For a class of languages \Cs, \Cs-separation takes two regular languages as input and asks whether there exists a third one in \Cs including the first one and disjoint from the second one. Settling whether separation is decidable for the levels within a given concatenation hierarchy is among the most fundamental and challenging questions in formal language theory. In all prominent cases, it is open, or answered positively for low levels only. Recently, a breakthrough was made using a generic approach for a specific kind of hierarchy: those with a \emph{finite} basis. In this case, separation is always decidable for levels 1/2, 1 and 3/2.

  Our main theorem is similar but independent: we consider hierarchies with possibly infinite bases, but that may only contain \emph{group languages}. An example is the group hierarchy of Pin and Margolis: its basis consists of all group languages. Another example is the quantifier alternation hierarchy of first-order logic with modular predicates (\fowm): its basis consists of the languages that count the length of words modulo some number. Using a generic approach, we show that for any such hierarchy, if separation is decidable for the basis, then it is decidable as well for levels 1/2, 1 and 3/2 (we actually solve a more general problem called covering). This complements the aforementioned result nicely: all bases considered in the literature are either finite or made of group languages. Thus, one may handle the lower levels of any prominent hierarchy in a \emph{generic way}.
\end{abstract}

\maketitle

\section{Introduction}
\label{sec:intro}
\noindent\textbf{Context.}
Concatenation hierarchies are natural classifications of regular languages. They were motivated by a celebrated theorem of Schützenberger~\cite{sfo}, McNaughton and Papert~\cite{mnpfo}, which is twofold. Its first part is simple, yet important: it states that the class of regular languages that can be defined from singletons using Boolean operations and concatenation only (but \emph{no} Kleene star) coincides with the class of languages that can be defined in  first-order logic. In other words, \emph{star-free} languages are exactly \emph{first-order definable} ones. The second and more difficult part consists of an algorithm that inputs a regular language, and outputs whether it belongs to this class.

This motivated Brzozowski and Cohen~\cite{BrzoDot} for introducing the \emph{dot-depth hierarchy}. It stratifies star-free languages in an infinitely increasing sequence of levels~\cite{BroKnaStrict} spanning the whole class of star-free languages. Intuitively, a level in this hierarchy captures the number of alternations between complement and concatenation that are needed to express a star-free language. The dot-depth rose to prominence following the work of Thomas~\cite{ThomEqu}, who proved an exact correspondence with the quantifier alternation hierarchy of first-order logic: each level in the dot-depth hierarchy consists of all languages that can be defined with a prescribed number of quantifier blocks.

One of the most famous open problems in automata theory is to settle whether the membership problem is decidable for each individual level: is there an algorithm deciding whether an input regular language belongs to this level?

The literature about this problem is rich. We refer the reader to the surveys~\cite{PZ:Siglog15,pzgenconcat,jep-dd45}. The bottom line is that after more than 45 years, little is known. Let us briefly survey the most important cornerstones in this line of research.

\smallskip\noindent\textbf{State of the art.}
The dot-depth is a particular instance of a \emph{concatenation hierarchy}.
All such hierarchies are built through a uniform construction scheme. Levels are numbered by integers 0,1,2... or half integers 1/2, 3/2, 5/2... Level~0 is a class of languages called the basis, which is the only parameter specific to the hierarchy. New levels are then built using two operations. The first is closure under marked concatenation $K,L\mapsto KaL$ (for some letter $a$) and union. This closure operator, when applied to level~$n\in\nat$, produces level~$n+1/2$. The second operator is closure under Boolean operations, which yields level $n+1$ out of level $n+1/2$.

For instance, the basis of the dot-depth hierarchy consists of 4 languages: the empty set, the singleton set consisting of the empty word, and their complements.
Between 1971 and 2015, membership was shown to be decidable up to level~5/2 and for level~7/2~\hbox{\cite{arfi87,knast83,gssig2,pzbpolc,pzqalt,pzsucc,pseps3,pseps3j}}.~The state of the art is the same for the \emph{Straubing-Thérien hierarchy}~\cite{StrauConcat,TheConcat}, whose basis is the empty set and the set of all words. Actually the above results for the dot-depth were obtained by reduction to this hierarchy, via transfer theorems~\cite{StrauVD,pzsucc}.

A crucial point is that all recent results are based on more general problems than membership: separation and covering. For instance, separation for a class \Cs takes as input \emph{two} regular languages, and asks whether there exists a third language, which belongs to \Cs, contains the first language, and is disjoint from the second. Membership is the particular case of separation when the input consists of a regular language and its complement.
Thus, deciding membership for a class \Cs reduces to deciding separation for \Cs, which itself reduces to deciding covering~for~\Cs.

\smallskip\noindent\textbf{A generic result.}
Recently, we proved that \emph{all} known results on both the Straubing-Thérien and the dot-depth hierarchies follow from a single one. This generic theorem~\cite{pzbpolc,pseps3j} states that in any hierarchy having a finite basis satisfying mild hypotheses, covering is decidable at levels 1/2, 1 and 3/2. All results stated above follow from this via simpler observations.

Establishing such generic results is of course desirable. This avoids developing specific arguments, pinpointing the key hypotheses  needed for obtaining covering or separation algorithms, and this is amenable to generalization.

\smallskip\noindent\textbf{Contribution.}
In the literature, there is another important kind of concatenation hierarchy: those with a basis made exclusively of group languages. A \emph{group language} is a language recognized by an automaton whose letters act as permutations on states. The most prominent representative is the \emph{group hierarchy}, whose basis consists of all group languages. Pin and Margolis~\cite{MargolisP85} proved that levels 1/2 and 1 have decidable membership. An algebraic attempt to go higher was proposed in \cite{pinbridges}. The idea was to reduce membership at a given level to covering at the same level in the Straubing-Thérien hierarchy. Unfortunately, while the connection between both hierarchies was established, it did not provide the desired algorithm. In this paper, we develop orthogonal, complementary ideas. Note that our aforementioned generic result does not apply here, since the class of group languages is infinite.

Our contribution is another \emph{\bfseries generic result} which nicely complements the one applying to finite bases. While both statements are similar, they are independent. Our main theorem states that \emph{covering} is decidable at levels 1/2, 1 and 3/2 in \emph{any} hierarchy whose basis consists of \emph{group languages}, provided that this basis has \emph{decidable separation}.

\smallskip\noindent\textbf{Applications.} Let us state three applications of this result. First, the class of all group languages is known to have decidable separation~(this follows from a result of Ash~\cite{Ash91}, which was connected to separation by Almeida~\cite{MR1709911}). Therefore, our generic statement applies to the group hierarchy of Pin and Margolis.

The other examples have a nice logical interpretation. Indeed, the result of Thomas~\cite{ThomEqu} was generalized in~\cite{pzgenconcat}: \emph{every} concatenation hierarchy corresponds, level by level, to the quantification alternation hierarchy within a variant of first-order logic over words. While the hierarchies differ by their bases, the variants differ by their set of available~predicates.

This correspondence makes it possible to present our second example as the quantifier alternation hierarchy of first-order logic with modular predicates (\fowm). It corresponds to the concatenation hierarchy whose basis consists of the languages that count the length of words modulo some number. It is simple to show that this basis consists of group languages and has decidable separation. Hence, our generic theorem applies: it strengthens and unifies results of~\cite{ChaubardPS06,KufleitnerW15}, which dealt with membership at levels up to 3/2, and~\cite{Zetzsche18}, which proved separation at level~1 with a combinatorial~proof leading to a brute force algorithm, orthogonal to our techniques.
Finally, when the basis consists of languages recognized by commutative groups, we obtain by~\cite{pzgenconcat,Eilenberg_book_B}  the quantifier alternation hierarchy of first-order logic endowed with predicates counting the number of occurrences of a letter before a position, modulo some integer. This basis has decidable covering~\cite{abelian_pt,MR1709911}, so our result applies to this hierarchy, which was not yet investigated.

\medskip\noindent
{\bf Organization.}
In Section~\ref{sec:prelims}, we recall the basic notions that we need. The main theorem is stated in Section~\ref{sec:hiera}. The framework that we use is recalled in Sections~\ref{sec:ratms} and~\ref{sec:units}. Finally, Sections~\ref{sec:polg},~\ref{sec:bpolg} and~\ref{sec:pbpolg} are devoted to presenting the algorithms of the three parts of the main theorem: the decidability of covering at levels 1/2, 1 and 3/2, respectively.

\noindent
Due to space limitations, some proofs are given in appendix.

\section{Preliminaries}
\label{sec:prelims}

\subsection{Classes of languages}

We fix an arbitrary finite alphabet $A$ for the whole paper. As usual, $A^*$ denotes the set of all words over $A$, including the empty word~\veps. We let $A^{+}=A^{*}\setminus\{\veps\}$. For $u,v \in A^*$, we write $u \cdot v$ or $uv$ the word obtained by concatenating $u$ and~$v$. Finally, for $w \in A^*$, we write $|w| \in \nat$ for its length.

A \emph{language} is a subset of $A^*$. We denote the singleton language $\{u\}$ by $u$. One may lift the concatenation operation to languages: for $K,L \subseteq A^*$, we let $KL = \{uv \mid u \in K \text{\;and\;} v \in L\}$. Additionally, we consider \emph{marked concatenation}. Given $K,L \subseteq A^*$, a marked concatenation of $K$ with $L$ is a language of the form $KaL$ for some letter $a \in A$.

A \emph{class of languages} \Cs is a set of languages. We shall work with robust classes satisfying standard closure properties:
\begin{itemize}
\item \Cs is a \emph{lattice} when it is closed under union and intersection, $\emptyset \in \Cs$ and $A^* \in \Cs$.
\item A \emph{Boolean algebra} is a lattice closed under complement.
\item \Cs is \emph{quotient-closed} when for every $L \in \Cs$ and $w \in A^*$, the following two languages belong to \Cs:
  \[
    \begin{array}{lll}
      w^{-1}L & \stackrel{\text{def}}= & \{u \in A^* \mid wu \in L\}, \\
      Lw^{-1} & \stackrel{\text{def}}= & \{u \in A^* \mid uw \in L\}.
    \end{array}
  \]
\end{itemize}
All classes considered in the paper are (at least) \pvaris. Furthermore, they  are included in the class of \emph{regular languages}. These are the languages that can be equivalently defined by monadic second-order logic, finite automata or finite monoids. We use the definition based on monoids.

\smallskip\noindent
{\bf Regular languages.} A \emph{semigroup} is a pair $(S,\cdot)$ where $S$ is a set and ``$\cdot$'' is an associative binary operation on $S$ (often called multiplication).  It is standard to abuse terminology and make the binary operation implicit: one simply says that ``$S$ is a semigroup''. A \emph{monoid} $M$ is a semigroup whose multiplication has a neutral element denoted by ``$1_M$''. Recall that an idempotent of a semigroup $S$ is an element $e \in S$ such that $ee = e$. A standard result in semigroup theory states that when $S$ is \emph{finite}, there exists $\omega(S) \in \nat$ (written $\omega$ when $S$ is understood) such that $s^\omega$ is idempotent for every $s \in S$.

Clearly, $A^*$ is a monoid whose multiplication is concatenation (\veps is the neutral element).  Hence, given a monoid $M$, we may consider morphisms $\alpha: A^* \rightarrow M$. We say that a language $L \subseteq A^*$ is \emph{recognized} by such a morphism $\alpha$ when there exists $F \subseteq M$ such that $L = \alpha\inv(F)$. It is well-known that the regular languages are exactly those which can be recognized by a morphism $\alpha: A^* \rightarrow M$ where $M$ is a \emph{finite} monoid.

\smallskip\noindent
{\bf Group languages.} A group is a monoid $G$ such that every element $g \in G$ has an inverse $g\inv \in G$, \emph{i.e.}, \hbox{$gg\inv = g\inv g = 1_G$}. We call ``\emph{group language}'' a language $L$ which is recognized by a morphism into a \emph{finite group}. In the paper, we consider classes that are \varis of group languages (\emph{i.e.}, containing group languages only).

\begin{example}\label{ex:allgroups}
  The most immediate example of \vari of group languages (which is also the largest one) is the class of \emph{all} group languages. We write it \grp.
\end{example}

\Varis of group languages are more general than the classes of group languages that are usually considered. Typically, publications on the topic consider \emph{varieties} of group languages, which involve an additional closure property called ``inverse morphic image'' (see~\cite{pinbridges} for details). Let us present a class which is a \vari of group languages, but not a variety.

\begin{example}\label{ex:mods}
  The class \md containing the finite Boolean combinations of languages $\{w \in A^* \mid |w| = k \mod m\}$ with $k,m \in \nat$ such that $k < m$, is a \vari of group languages. We use it as a running example.
\end{example}

\subsection{Separation and covering}

In the paper, we use two decision problems to investigate specific classes of languages (all built from \varis of group languages): separation and covering. We define them here. The former is standard while the latter was introduced in~\cite{pzcovering2}. Both of them are parametrized by an arbitrary class of languages \Cs. We start with separation.

\medskip
\noindent
{\bf Separation.} Given three languages $K,L_1,L_2$, we say that $K$ \emph{separates} $L_1$ from $L_2$ if $L_1 \subseteq K \text{ and } L_2 \cap K = \emptyset$. Given a class of languages \Cs, we say that $L_1$ is \emph{\Cs-separable} from $L_2$ if some language in \Cs separates $L_1$ from $L_2$. Observe that when \Cs is not closed under complement, the definition is not symmetrical: $L_1$ could be \Cs-separable from $L_2$ while $L_2$ is not \Cs-separable from $L_1$. The separation problem associated to a given class \Cs is as follows:

\medskip

\begin{tabular}{rl}
  {\bf INPUT:}  &  Two regular languages $L_1$ and $L_2$. \\
  {\bf OUTPUT:} &  Is $L_1$ \Cs-separable from $L_2$?
\end{tabular}

\begin{remark}
  Separation generalizes the classical membership problem which asks whether a single regular language belongs to \Cs. Indeed, $L \in \Cs$ if and only if $L$ is \Cs-separable from $A^* \setminus L$.
\end{remark}

\noindent
{\bf Covering.} This more general problem was introduced in~\cite{pzcovering2}. Given a language $L$, a \emph{cover of $L$} is a \emph{\bf finite} set of languages \Kb such that $L \subseteq \bigcup_{K \in \Kb} K$. Moreover, if \Cs is a class, a \Cs-cover of $L$  is a cover \Kb of $L$ such that all $K \in \Kb$ belong to \Cs.

Covering takes as input a language $L_1$ and a \emph{finite set of languages} $\Lb_2$. A \emph{separating cover} for the pair $(L_1,\Lb_2)$ is a cover \Kb of $L_1$ such that for every $K\in\Kb$, there exists $L \in \Lb_2$ which satisfies $K \cap L = \emptyset$. Finally, given a class \Cs, we say that the pair $(L_1,\Lb_2)$ is \Cs-coverable when there exists a separating \Cs-cover. The \Cs-covering problem is now defined as follows:

\medskip

\begin{tabular}{rl}
  {\bf INPUT:}  &  A regular language $L_1$ and\\
                & a finite set of regular languages $\Lb_2$.\\
  {\bf OUTPUT:} &  Is $(L_1,\Lb_2)$ \Cs-coverable?
\end{tabular}

\medskip

It is straightforward to prove that covering generalizes separation (provided that the class \Cs is a lattice) as stated in the following lemma (see Theorem~3.5 in~\cite{pzcovering2} for the proof).

\begin{lemma}\label{lem:septocove}
  Let \Cs be a lattice and $L_1,L_2 \subseteq A^*$. Then $L_1$ is \Cs-separable from $L_2$, if and only if $(L_1,\{L_2\})$ is \Cs-coverable.
\end{lemma}

\section{Concatenation hierarchies and main theorem}
\label{sec:hiera}
In this section, we present the particular classes that we investigate in the paper and outline our results.

\subsection{Closure operations and concatenation hierarchies}

We are interested in \emph{concatenation hierarchies}. We briefly recall this notion (see~\cite{jep-dd45,pzgenconcat} for details). A concatenation hierarchy is an increasing sequence of classes of languages, which depends on a single parameter: an arbitrary \vari, called its \emph{basis}. Once the basis is fixed, the construction is uniform. Languages are classified into levels: each new level is built by applying one of two generic operations to the previous one. Let us define these~operations.

\smallskip
\noindent
{\bf Boolean closure.} Given a class \Cs, its \emph{Boolean closure}, denoted by \bool\Cs is the least Boolean algebra containing \Cs. The following lemma is immediate from the definitions (it holds simply because quotients commute with Boolean operations).

\begin{lemma}\label{lem:boolclos}
  Let \Cs be a \pvari. Then \bool{\Cs} is a \vari.
\end{lemma}

\smallskip
\noindent
{\bf Polynomial closure.} Given a class \Cs,  the \emph{polynomial closure} of \Cs, denoted by \pol{\Cs}, is the least class containing \Cs which is closed under both union and marked concatenation: for every $K,L \in \pol{\Cs}$ and $a \in A$, we have $K \cup L \in \pol{\Cs}$ and $KaL \in \pol{\Cs}$. While this is not obvious from the definition, when the input class \Cs is a \pvari, its polynomial closure \pol{\Cs} remains a \pvari (the difficulty is to prove closure under intersection). This was originally proved by Arfi~\cite{arfi87} (see~\cite{jep-intersectPOL,pzgenconcat} for recent proofs).

\begin{theorem}[Arfi~\cite{arfi87}]\label{thm:polclos}
  Let \Cs be a \pvari. Then, \pol{\Cs} is a \pvari closed under concatenation and marked concatenation.
\end{theorem}

\smallskip
\noindent
{\bf Concatenation hierarchies.} We let \Cs as a \vari. The \emph{concatenation hierarchy of basis \Cs} classifies languages into levels of two kinds: full levels (denoted by 0, 1, 2,...) and half levels (denoted by 1/2, 3/2, 5/2,...):
\begin{itemize}
\item Level $0$ is the basis \Cs.
\item Each \emph{half level} $n+\frac{1}{2}$, for $n\in\nat$, is the \emph{polynomial closure} of the previous full level, \emph{i.e.}, of level $n$.
\item Each \emph{full level} $n+1$, for $n\in\nat$, is the \emph{Boolean closure} of the previous half level, \emph{i.e.}, of level $n+\frac12$.
\end{itemize}

In view of Lemma~\ref{lem:boolclos} and Theorem~\ref{thm:polclos}, it is immediate that every full level is a \vari while every half level is a \pvari. Moreover, we have the following useful fact (see~\cite[Lemma~3.4]{pseps3j} for a proof).

\begin{fact}\label{fct:finitelangs}
  Let \Ds be a level greater or equal to $1$ in some concatenation hierarchy. Then $\{w\} \in \Ds$ for every $w \in A^*$.
\end{fact}

In this paper, we are mainly interested in levels 1/2,~1 and 3/2. By definition, they correspond to the classes \pol{\Cs}, \bool{\pol{\Cs}} and \pol{\bool{\pol{\Cs}}}. For the sake of avoiding clutter, we shall write \bpol{\Cs} for \bool{\pol{\Cs}} and \pbpol{\Cs} for \pol{\bool{\pol{\Cs}}}. Our main theorem applies to those levels for every basis \Cs which is a \emph{\vari of group languages}. It is as follows.

\begin{theorem}[Main result]\label{thm:main}
  Consider a concatenation hierarchy whose basis \Cs contains only \textbf{group languages} and such that \Cs-separation is decidable. Then, covering and separation are decidable for levels 1/2, 1 and 3/2 of this hierarchy.
\end{theorem}

Theorem~\ref{thm:main} is generic and applies to all ``group based'' hierarchies. It complements similar results which were recently proved in~\cite{pseps3,pseps3j} (for levels 1/2 and 3/2) and~\cite{pzboolpol,pzbpolc} (for level 1) for \emph{finitely based} hierarchies.

\begin{theorem}[\cite{pseps3j,pzbpolc}]\label{thm:fbasis}
  Consider a concatenation hierarchy whose basis is \textbf{finite}. Then, covering and separation are decidable for levels 1/2, 1 and 3/2.
\end{theorem}

All prominent concatenation hierarchies investigated in the literature have a basis which is either finite or made of group languages. Hence, when put together, the two above theorems can be used to handle the lower levels of all prominent concatenation hierarchies in a \emph{generic way}.

Additionally, one may combine Theorem~\ref{thm:main} with a result of~\cite{pzgenconcat} to get information on level 5/2 of group based hierarchies. It is shown in~\cite{pzgenconcat} that given an arbitrary hierarchy, if separation is decidable for some half level, then so is the \emph{membership problem} for the next half level. Hence, we get the following corollary of Theorem~\ref{thm:main}.

\begin{corollary}\label{cor:main}
  Consider a concatenation hierarchy whose basis \Cs contains only \textbf{group languages} and such that \Cs-separation is decidable. Then, given as input a regular language $L$, one may decide whether $L$ belongs to level 5/2 of this hierarchy.
\end{corollary}

The proof of Theorem~\ref{thm:main} spans the remaining sections of the paper. Each level mentioned in the statement is handled independently. In Section~\ref{sec:ratms}, we recall a framework designed to handle the covering problem. It was originally introduced in~\cite{pzcovering2} and we use it to obtain all results announced in Theorem~\ref{thm:main}. We apply it for level 1/2 in Section~\ref{sec:polg}. Handling levels 1 and 3/2 is much more involved and requires extending the framework. We do so in Section~\ref{sec:units} and handle these levels in Sections~\ref{sec:bpolg} and~\ref{sec:pbpolg} respectively.

However, let us first illustrate Theorem~\ref{thm:main} by presenting two prominent concatenation hierarchies to which it applies.

\subsection{Important applications of Theorem~\ref{thm:main}}

We present two prominent examples of concatenation hierarchies whose bases are \vari of group languages (we already presented these bases in Section~\ref{sec:prelims}).

\medskip
\noindent
{\bf Group hierarchy.} This is the hierarchy of basis \grp (the class of \emph{all} group languages). It was originally introduced by Margolis and Pin~\cite{MargolisP85}. To our knowledge, the following results were known for this hierarchy:
\begin{itemize}
\item Separation is decidable for the basis \grp. This is a corollary of a theorem by Ash~\cite{Ash91} which solves a longstanding conjecture in algebra (Rhodes' type~II conjecture). We refer the reader to~\cite{HenckellMPR91} for background on this question.
\item Membership is decidable for levels 1/2 and 1. The former follows from the decidability of separation for \grp and a transfer result of~\cite{pzgenconcat}. The latter is proved in~\cite{MargolisP85}.
\end{itemize}
Since separation is decidable for the basis \grp, we obtain the following result from Theorem~\ref{thm:main} and Corollary~\ref{cor:main}.

\begin{corollary}\label{cor:group}
  Separation and covering are decidable for levels 1/2, 1 and 3/2 in the group hierarchy. Moreover, membership is decidable for level 5/2.
\end{corollary}

\noindent
{\bf The hierarchy of basis \md.} This example is quite important because it has a natural \emph{logical} equivalent. It is known that \emph{every} concatenation hierarchy corresponds to the \emph{quantifier alternation hierarchy} within a particular variant of first-order logic over words (the variants differ by the set of predicates that are allowed in sentences). This was first observed by Thomas~\cite{ThomEqu} for a particular example: the dot-depth hierarchy of Brzozowski and Cohen~\cite{BrzoDot} (whose basis is the finite class $\{\emptyset,\{\veps\},A^+,A^*\}$). However, the ideas of Thomas can be generalized to all concatenation hierarchies (see~\cite{pzgenconcat}).

Of course, depending on the concatenation hierarchy, the variant of first-order logic one ends up with may or may not be natural. It turns out that for the hierarchy of basis \md, we get a \emph{standard} one: first-order logic with modular predicates (\fowm). One may view a finite word $w$ as a logical structure made of a sequence of positions numbered from $0$ to $|w|-1$. Each position can be quantified and carries a label in~$A$. We denote by \fowm the variant of first-order logic equipped with the following predicates:
\begin{itemize}
\item For each $a \in A$, a unary predicate $P_a$ selecting positions labeled with letter ``$a$''.
\item A binary predicate ``$<$'' interpreted as the linear order.
\item For each integers $0\leq i < d$, a unary predicate $\mathsf{MOD}^d_i$ selecting positions that are congruent to $i$ modulo $d$.
\item A constant $\mathsf{D}^d_i$, which holds for words whose length is congruent to $i$ modulo $d$.
\end{itemize}
A sentence of \fowm defines a language (it consists of all words satisfying the sentence). Thus, \fowm defines a \emph{class of languages} which we also denote by \fowm.

We are not interested in \fowm itself: we consider its quantifier alternation hierarchy. One may classify sentences of \fowm by counting their number of quantifier alternations. Let $n \in \nat$. A sentence is $\siwm{n}$, if it can be rewritten into a sentence in prenex normal form which has either, exactly $n$ blocks of quantifiers starting with an $\exists$ or \emph{strictly less} than $n$ blocks of quantifiers. For example, consider the following sentence (already in prenex normal form)
\[
  \exists x_1 \exists x_2 \forall x_3 \exists x_4
  \ \varphi(x_1,x_2,x_3,x_4) \quad \text{(with $\varphi$ quantifier-free)}
\]
\noindent
This sentence is \siwmt. In general, the negation of a \siwm{n} sentence is not a \siwm{n} sentence. Hence it is relevant to define \bswm{n} sentences as the Boolean combinations of \siwm{n} sentences. This gives a hierarchy of classes of languages: for every $n \in \nat$, we have $\siwm{n} \subseteq \bswm{n} \subseteq \siwm{n+1}$.

This hierarchy has been widely investigated in the literature. For membership, it was known that the problem is decidable for \bswmu (see~\cite{ChaubardPS06}) as well as \siwmu and \siwmd (see~\cite{KufleitnerW15}). Furthermore, it was recently shown that separation is decidable for \bswmu~\cite{Zetzsche18}. We are able to reprove these results, generalize them to the covering problem and push them to \siwmt. Indeed, the generic correspondence of~\cite{pzgenconcat} implies that for every $n \in \nat$, level $n$ in the hierarchy of basis \md corresponds to \bswm{n} while level $n+\frac{1}{2}$ corresponds to \siwm{n}. Moreover, proving the decidability of \md-separation is a simple exercise (we present a proof in appendix). Consequently, we obtain the following statement as a corollary of Theorem~\ref{thm:main}.

\begin{corollary}\label{cor:mod}
  Separation and covering are decidable for the logics \siwmu,\bswmu and \siwmd. Moreover, membership is decidable for \siwmt.
\end{corollary}

\section{Framework}
\label{sec:ratms}
In this section, we present the framework which we shall use to obtain the three decidability results announced in Theorem~\ref{thm:main}. It was originally introduced in~\cite{pzcovering2}. We refer the reader to this paper for details and the proofs of the statements.

\subsection{\Ratms}

The framework is based on an algebraic object called ``\ratm''. They are morphisms of  commutative and idempotent monoids. We write such monoids $(R,+)$: we call the binary operation  ``$+$'' \emph{addition} and denote the neutral element by $0_R$. Being idempotent means that for all $r \in R$, we have $r + r = r$. For every commutative and idempotent monoid $(R,+)$, one may define a canonical ordering $\leq$ over $R$:
\[\text{For all }
  r, s\in R,\quad r\leq s \text{ when } r+s=s.
\]
One may verify that $\leq$ is a partial order which is compatible with addition. Moreover, every morphism between commutative and idempotent monoids is increasing for this ordering.

\begin{example}
  For every set $E$, $(2^E,\cup)$ is an idempotent and commutative monoid. The neutral element is $\emptyset$ and the canonical ordering is inclusion.
\end{example}

When dealing with subsets of a commutative and idempotent monoid $(R,+)$, we shall often apply a \emph{downset operation}. Given $S \subseteq R$, we write:
\[
  \dclosr S = \{r \in R \mid r \leq s \text{ for some $s \in S$}\}.
\]
We extend this notation to Cartesian products of arbitrary sets with $R$. Given some set $X$ and $S \subseteq X \times R$, we write,
\[
  \dclosr S = \{(x,r) \in X \times R \mid \text{$\exists s \in R$ s.t. $r \leq s$ and $(x,s) \in S$}\}
\]
We may now define \ratms. A \ratm is a morphism $\rho: (2^{A^*},\cup) \to (R,+)$ where $(R,+)$ is a \emph{finite} idempotent and commutative monoid, called the \emph{rating set of $\rho$}. That is, $\rho$ is a map from $2^{A^{*}}$ to $R$ satisfying the following properties:
\begin{enumerate}[resume,label=(\arabic*)]
\item\label{itm:bgen:fzer} $\rho(\emptyset) = 0_R$.
\item\label{itm:bgen:ford} For all $K_1,K_2 \subseteq A^*$, $\rho(K_1\cup K_2)=\rho(K_1)+\rho(K_2)$.
\end{enumerate}

For the sake of improved readability, when applying a \ratm $\rho$ to a singleton set $\{w\}$, we shall write $\rho(w)$ for $\rho(\{w\})$. Additionally, we write $\rho_*: A^* \to R$ for the restriction of $\rho$ to $A^*$: for every $w \in A^*$, we have $\rho_*(w) = \rho(w)$ (this notation is useful when referring to the language $\rho_*\inv(r) \subseteq A^*$, which consists of all words $w \in A^*$ such that $\rho(w) = r$).

\smallskip

Most of the theory makes sense for arbitrary \ratms. However, we shall often have to work with special \ratms satisfying additional properties. We define two kinds.

\smallskip
\noindent
{\bf \Nice \ratms.} We say a \ratm $\rho: 2^{A^*} \to R$ is \nice when, for every language $K \subseteq A^*$, there exists finitely many words $w_1,\dots,w_n \in K$ such that $\rho(K) = \rho(w_1) + \cdots + \rho(w_k)$.

When a \ratm $\rho: 2^{A^*} \to R$ is \nice, it is characterized by the canonical map $\rho_*: A^* \to R$. Indeed, for $K \subseteq A^*$, we may consider the sum of all elements $\rho(w)$ for $w \in K$: while it may be infinite, this sum boils down to a finite one since $R$ is commutative and idempotent. The hypothesis that $\rho$ is \nice implies that $\rho(K)$ is equal to this sum.

\smallskip
\noindent
{\bf \Mratms.} A \ratm $\rho: 2^{A^*} \to R$ is \tame when its rating set $R$ has more structure: it needs to be an \emph{idempotent semiring}. Moreover, $\rho$ has to satisfy an additional property connecting this structure to language concatenation. Namely, it has to be a morphism of semirings.

A \emph{semiring} is a tuple $(R,+,\cdot)$ where $R$ is a set and ``$+$'' and ``$\cdot$''  are two binary operations called addition and multiplication, such that the following axioms are satisfied:
\begin{itemize}
\item $(R,+)$ is a commutative monoid.
\item $(R,\cdot)$ is a monoid (the neutral element is denoted by $1_R$).
\item Multiplication distributes over addition. For $r,s,t \in R$, $r \cdot (s + t) = (r \cdot s) + (r \cdot t)$  and $(r + s) \cdot t = (r \cdot t) + (s \cdot t)$.
\item The neutral element ``$0_R$'' of $(R,+)$ is a zero for $(R,\cdot)$: $0_R \cdot r = r \cdot 0_R = 0_R$ for every $r \in R$.
\end{itemize}
A semiring $R$ is \emph{idempotent} when $r + r = r$ for every $r \in R$, \emph{i.e.}, when the additive monoid $(R,+)$ is idempotent (there is no additional constraint on the multiplicative monoid $(R,\cdot)$).

\begin{example}\label{ex:bgen:semiring}
  A key example of infinite idempotent semiring is the set $2^{A^*}$. Union is the addition and language concatenation is the multiplication (with $\{\varepsilon\}$ as neutral element).
\end{example}

Clearly, any finite idempotent semiring $(R,+,\cdot)$ is in particular a rating set: $(R,+)$ is an idempotent and commutative monoid. In particular, one may verify that the canonical ordering ``$\leq$'' on $R$ is compatible with multiplication as well.

We may now define \mratms: as expected they are semiring morphisms. Let $\rho: 2^{A^*} \to R$ be a \ratm: $(R,+)$ is an idempotent commutative monoid and $\rho$ is a morphism from $(2^{A^*},\cup)$ to $(R,+)$. We say that $\rho$ is \tame when the rating set $R$ is equipped with a multiplication ``$\cdot$'' such that $(R,+,\cdot)$ is an idempotent semiring and $\rho$ is also a monoid morphism from $(2^{A^*},\cdot)$ to $(R,\cdot)$. That is, the two following additional axioms have to be satisfied:
\begin{enumerate}[resume,label=(\arabic*)]
\item\label{itm:bgen:funit} $\rho(\varepsilon) = 1_R$.
\item\label{itm:bgen:fmult} For all $K_1,K_2 \subseteq A^*$, we have $\rho(K_1K_2) = \rho(K_1) \cdot \rho(K_2)$.
\end{enumerate}
Altogether, this exactly says that $\rho$ must be a semiring morphism from $(2^{A^*},\cup,\cdot)$ to $(R,+,\cdot)$.

\begin{remark}
  The \ratms which are both \nice and \tame are finitely representable. Indeed, as we explained above, if a \ratm $\rho: 2^{A^*} \to R$ is \nice, it is characterized by the canonical map $\rho_*: A^* \to R$.  When $\rho$ is additionally \tame, $\rho_*$ is finitely representable since it is a morphism into a finite monoid. Hence, we may speak about algorithms taking \nice \mratms as input.

  The \ratms which are not \nice and \tame cannot be finitely represented in general. However, they remain crucial in the paper: while our main statements consider \nice \mratms, many proofs involve auxiliary \ratms which are neither \nice nor \tame.
\end{remark}

\subsection{Optimal \imprints}

Now that we have \ratms, we turn to \imprints. Consider a \ratm $\rho: 2^{A^*} \to R$. Given any finite set of languages~\Kb, we define the $\rho$-\imprint of~\Kb. Intuitively, when \Kb is a cover of some language $L$, this object measures the ``quality'' of \Kb. The $\rho$-\imprint \emph{of \Kb} is the following subset of~$R$:
\[
  \prin{\rho}{\Kb} = \dclosr \{\rho(K) \mid K \in\Kb\}.
\]
We may now define optimality. Consider an arbitrary \ratm $\rho: 2^{A^*} \to R$ and a lattice \Cs. Given a language $L$, an optimal \Cs-cover of $L$ for $\rho$ is a \Cs-cover \Kb of $L$ which satisfies the following property:
\[
  \prin{\rho}{\Kb} \subseteq \prin{\rho}{\Kb'} \quad \text{for every \Cs-cover $\Kb'$ of $L$}.
\]
In general, there can be infinitely many optimal \Cs-covers for a given \ratm $\rho$. The key point is that there always exists at least one (this requires the hypothesis that \Cs is a lattice).

\begin{lemma} \label{lem:bgen:opt}
  Let \Cs be a lattice. For every language $L$ and every \ratm $\rho$, there exists an optimal \Cs-cover of $L$ for $\rho$.
\end{lemma}

Clearly, for a lattice \Cs, a language $L$ and a \ratm $\rho$, all optimal \Cs-covers of $L$ for $\rho$ have the same $\rho$-\imprint. Hence, this unique $\rho$-\imprint is a \emph{canonical} object for \Cs, $L$ and $\rho$. We call it the \emph{\Cs-optimal $\rho$-\imprint on $L$} and we write it $\opti{\Cs}{L,\rho}$:
\[
  \opti{\Cs}{L,\rho} = \prin{\rho}{\Kb} \quad \text{for any optimal \Cs-cover \Kb of $L$  for $\rho$}.
\]
An important special case is when $L = A^*$. In that case, we write \copti{\rho} for \copti{A^*,\rho}. Finally, we have the following useful fact which is immediate from the definitions.

\begin{fact} \label{fct:linclus}
  Let $\rho$ be a \ratm and consider two languages $H,L$ such that $H \subseteq L$. Then, $\opti{\Cs}{H,\rho} \subseteq \opti{\Cs}{L,\rho}$.
\end{fact}

\subsection{Connection with covering}

We may now connect these definitions to the covering problem. The key idea is that solving \Cs-covering boils down to computing \Cs-optimal \imprints from input \nice \mratms. There are actually two statements (both taken from~\cite{pzcovering2}). The first one is simpler but it only applies to classes \Cs which are Boolean algebras while the second (more involved) one applies to all lattices. We start with the former.

\begin{proposition} \label{prop:breduc}
  Let \Cs be a Boolean algebra. Assume that there exists an algorithm which computes \copti{\rho} from an input \nice \mratm $\rho$. Then, \Cs-covering is decidable.
\end{proposition}

In order to handle all lattices, one needs to consider several optimal \imprints simultaneously. This is formalized by the following additional object. Consider a  lattice \Cs, a morphism $\alpha: A^* \to M$ into a finite monoid $M$ and a \ratm $\rho: 2^{A^*} \to R$, we define the following subset of $M \times R$:
\[
  \popti{\Cs}{\alpha}{\rho} = \{(s,r) \in M \times R \mid r \in \opti{\Cs}{\alpha\inv(s),\rho}\}.
\]
We call \popti{\Cs}{\alpha}{\rho} the \emph{$\alpha$-pointed \Cs-optimal $\rho$-\imprint}. Clearly, \popti{\Cs}{\alpha}{\rho} encodes all the sets \copti{\alpha\inv(s),\rho} for $s \in M$.

\begin{proposition} \label{prop:lreduc}
  Let \Cs be a lattice. Assume that there exists an algorithm which computes \popti{\Cs}{\alpha}{\rho} from an input morphism $\alpha$ and an input \nice \mratm $\rho$. Then, \Cs-covering is decidable.
\end{proposition}

\section{Polynomial closure}
\label{sec:polg}
We prove the first part of Theorem~\ref{thm:main}: for every \vari of group languages \Cs, if \Cs-separation is decidable, then so is \pol{\Cs}-covering.

We use \ratms: for every morphism $\alpha: A^* \to M$ and \nice \mratm $\rho: 2^{A^*} \to R$, we characterize $\pocopti$. When \Cs-separation is decidable, it is simple to deduce a least fixpoint algorithm for computing \pocopti. Proposition~\ref{prop:lreduc} then yields that \pol{\Cs}-covering is decidable.

\begin{remark}
  Our characterization of \pocopti does not actually require $\rho$ to be \nice. While useless for deciding \pol{\Cs}-covering, this is significant. This will later be a requirement when handling \bpol{\Cs} and \pbpol{\Cs} (see remark~\ref{rem:bpol}).
\end{remark}

Before we can present the characterization, we require some terminology. We define a new notion for \ratms: optimal \idens. They are specifically designed to handle covering for concatenation hierarchies whose bases are made of group languages (it is through this notion that our characterizations are parameterized by \Cs-separation).

\subsection{Optimal \idens}

When handling \pol{\Cs} for a \vari of group languages \Cs (as well as \bpol{\Cs} and \pbpol{\Cs} later), we shall often encounter optimal \Cs-covers of the singleton $\{\veps\}$ for various \ratms. The definitions presented here are based on a single key idea: there always exists an optimal \Cs-cover of $\{\veps\}$ which consists of a single language.

\begin{remark}
  The definitions make sense for any lattice \Cs. However, they are mainly relevant when \Cs is a class of group languages. In practice, the other important lattices contain the singleton $\{\veps\}$ and $\{\{\veps\}\}$ is always an optimal \Cs-cover of $\{\veps\}$.
\end{remark}

Let \Cs be a lattice and $\rho: 2^{A^*} \to R$ be a \ratm. A \emph{\Cs-optimal \iden for $\rho$} is a language $L \in \Cs$ such that $\veps \in L$ and,
\[
  \rho(L) \leq \rho(L') \quad \text{for every $L' \in \Cs$ such that $\veps \in L'$}.
\]

\begin{remark}
  Since $\{\veps\}$ is a singleton, one may show that $L$ is a \Cs-optimal \iden for $\rho$ if and only if $\{L\}$ is an optimal \Cs-cover of $\{\veps\}$~for~$\rho$.
\end{remark}

As expected, there always exists a \Cs-optimal \iden for any \ratm $\rho$, provided that \Cs is a lattice.

\begin{lemma}\label{lem:epswit}
  For any lattice \Cs and any \ratm $\rho: 2^{A^*} \to R$, there exists a \Cs-optimal \iden for $\rho$.
\end{lemma}

We complete the definition with a key remark. Clearly, all \Cs-optimal \idens for $\rho$ have the same image under $\rho$. This is a canonical object for \Cs and $\rho$. We write it $\ioptic{\rho} \in R$:
\[
  \ioptic{\rho} = \rho(L) \quad \text{for any \Cs-optimal \iden $L$ for $\rho$}.
\]
We turn to a crucial property of \idens: when $\rho$ is a \emph{\nice \ratm}, computing \iopti{\Cs}{\rho} boils down to separation.

\begin{lemma}\label{lem:sepepswit}
  Let \Cs be a lattice and $\rho: 2^{A^*} \to R$ be a \nice \ratm. Then, $\iopti{\Cs}{\rho}$ is the sum of all $r \in R$ such that $\{\veps\}$ is not \Cs-separable from $\rho_*\inv(r)$.
\end{lemma}

Lemma~\ref{lem:sepepswit} has an important corollary. When $\rho: 2^{A^*} \to R$ is a \nice \mratm, the languages $\rho_*\inv(r)$ are regular (they are recognized by $\rho_*$). Hence, given a lattice \Cs, if \Cs-separation is decidable, we may compute \ioptic{\rho} from an input \nice \mratm.

\begin{corollary}\label{cor:epswit2}
  Consider a lattice \Cs such that \Cs-separation is decidable. Then, given as input a \nice \mratm $\rho: 2^{A^*} \to R$, one may compute the subset $\ioptic{\rho}$ of $R$.
\end{corollary}

Corollary~\ref{cor:epswit2} is crucial for our approach: this is exactly how the algorithms for \pol{\Cs}-, \bpol{\Cs}- and \pbpol{\Cs}-covering announced in Theorem~\ref{thm:main} are parametrized by \Cs-separation.

Let us point out that when considering a \emph{specific} lattice \Cs, it is usually possible to write an algorithm for computing \ioptic{\rho} from an input \nice \mratm which \emph{does not} involve separation. We present an example for the class~\md.

\begin{lemma}\label{lem:mdiopti}
  Let $\rho: 2^{A^*} \to R$ be a \nice \mratm. Then, $\iopti{\md}{\rho} = (\rho(A))^\omega + \rho(\veps)$.
\end{lemma}

Finally, we have the following lemma, which considers the special case when the lattice \Cs is a \vari of group languages. We shall need this lemma for proving the three main results of the paper.

\begin{lemma}\label{lem:polctoc}
  Let \Cs be a \vari of group languages and $\rho: 2^{A^*} \to R$ be a \ratm. Then, $\ioptic{\rho} \in \opti{\pol{\Cs}}{\{\veps\},\rho}$.
\end{lemma}

\begin{proof}
  Let \Kb be an optimal \pol{\Cs}-cover of $\{\veps\}$. By definition, there exists $K \in \Kb$ such that $\veps \in K$. Since $K \in \pol{\Cs}$, it is immediate from a simple induction on the definition of polynomial closure that we have $H \in \Cs$ such that $\veps \in H$ and $H \subseteq K$ (the key idea is that the marked concatenation of two languages cannot contain $\veps$). By definition of \ioptic{\rho}, we have $\ioptic{\rho} \leq \rho(H)$. Consequently, $\ioptic{\rho} \leq \rho(K)$, which yields $\ioptic{\rho} \in \prin{\rho}{\Kb}$. Since \Kb is an optimal \pol{\Cs}-cover of $\{\veps\}$, we obtain $\ioptic{\rho} \in \opti{\pol{\Cs}}{\{\veps\},\rho}$, as desired.
\end{proof}

\subsection{Characterization}

We characterize \pocopti when \Cs is a \vari of group languages. Let \Cs be such a class.

We first state the characterization. Consider some morphism $\alpha: A^* \to M$ and some \mratm $\rho: 2^{A^*} \to R$. Consider a subset $S \subseteq M \times R$. We say that $S$ is \emph{\pol{\Cs}-complete} for $\alpha$ and $\rho$ when it satisfies the following properties:
\begin{itemize}
\item \emph{\bfseries Trivial elements}: For all $w \in A^*$, $(\alpha(w),\rho(w)) \in S$.
\item \emph{\bfseries Downset}: We have $\dclosr S = S$.
\item \emph{\bfseries Multiplication}: For all, $(s,q),(t,r) \in S$, $(st,qr) \in S$.
\item\label{op:half:gpolclos} \emph{\bfseries \Cs-operation}: We have $(1_M,\ioptic{\rho}) \in S$.
\end{itemize}

We may now state the main theorem of the section: the least \pol{\Cs}-complete subset of $M \times R$ with respect to inclusion is exactly \pocopti (recall that \Cs is required to be a \vari of group languages).

\begin{theorem}\label{thm:half:gmainpolc}
  Consider a morphism $\alpha: A^* \to M$ and a \mratm $\rho:2^{A^*} \to R$. Then, \pocopti is the least \pol{\Cs}-complete subset of $M \times R$.
\end{theorem}

By Theorem~\ref{thm:half:gmainpolc}, if \Cs-separation is decidable, then one may compute \pocopti when given a morphism $\alpha: A^* \to M$ and a \nice \mratm $\rho:2^{A^*} \to R$ as input. Indeed, computing the least \pol{\Cs}-complete subset of $M \times R$ is achieved by an obvious least fixpoint procedure. We are able to implement \Cs-operation, since we may compute $\ioptic{\rho} \in R$: by Corollary~\ref{cor:epswit2}, this boils down to deciding \Cs-separation.

Together with Proposition~\ref{prop:lreduc}, this yields that when \Cs-separation is decidable, so is \pol{\Cs}-covering. Therefore, we get the first part of Theorem~\ref{thm:main}, regarding level~1/2. We conclude the section with the proof of Theorem~\ref{thm:half:gmainpolc}.

\subsection{Proof of Theorem~\ref{thm:half:gmainpolc}}

We fix a morphism $\alpha: A^* \to M$ and a \mratm $\rho: 2^{A^*} \to R$ for the proof. We need to show that \pocopti is the least \pol{\Cs}-complete subset of $M \times R$. We first prove that it is \pol{\Cs}-complete. This means that the least fixpoint procedure computing \pocopti is sound.

\smallskip
\noindent
{\bf Soundness.} We have to show that \pocopti satisfies the four properties in the definition of \pol{\Cs}-complete subsets. The first three (trivial elements, downset and multiplication) are generic: they are satisfied by \popti{\Ds}{\alpha}{\rho} for any \pvari \Ds (we refer the reader to~\cite{pzcovering2} for the proof). Hence, they are satisfied by \pocopti since \pol{\Cs} is a \pvari by Theorem~\ref{thm:polclos}. We concentrate on \pol{\Cs}-operation.

We show that $(1_M,\ioptic{\rho}) \in \pocopti$. In other words, we prove $\ioptic{\rho} \in \opti{\pol{\Cs}}{\alpha\inv(1_M),\rho}$. By Lemma~\ref{lem:polctoc}, we know that $\ioptic{\rho} \in \opti{\pol{\Cs}}{\{\veps\},\rho}$. Moreover, $\{\veps\} \subseteq \alpha\inv(1_M)$ since $\alpha$ is a morphism. Thus, Fact~\ref{fct:linclus} yields that $\ioptic{\rho} \in \opti{\pol{\Cs}}{\alpha\inv(1_M),\rho}$, finishing the soundness proof.

\smallskip
\noindent
{\bf Completeness.} We now prove that \pocopti is included in every \pol{\Cs}-complete subset of $M \times R$. This direction corresponds to completeness of the least fixpoint procedure computing \pocopti. For the proof, we fix $S \subseteq M \times R$ which is \pol{\Cs}-complete subset and show that $\pocopti \subseteq S$. The argument is based on the following proposition

\begin{proposition}\label{prop:polc}
  Let $t \in M$. There exists a \pol{\Cs}-cover $\Kb_t$ of $\alpha\inv(t)$ such that for every $K \in \Kb_t$, we have $(t,\rho(K)) \in S$.
\end{proposition}

Let us first use Proposition~\ref{prop:polc} to show that $\pocopti \subseteq S$. Let $(t,r) \in \pocopti$. We show that $(t,r) \in S$. By definition, we have $r \in \opti{\pol{\Cs}}{\alpha\inv(t),\rho}$. Consider the \pol{\Cs}-cover $\Kb_t$ of $\alpha\inv(t)$ given by Proposition~\ref{prop:polc}. By hypothesis, we have $r \in \prin{\rho}{\Kb_t}$. Thus, we get $K \in \Kb_t$ such that $r \leq \rho(K)$. Finally, since $(t,\rho(K)) \in S$ and $S = \dclosr S$ ($S$ is \pol{\Cs}-complete), we obtain that $(t,r) \in S$, finishing the proof.

It remains to prove Proposition~\ref{prop:polc}. We rely on the following key technical result.

\begin{lemma}\label{lem:half:pump}
  Let $H$ be a regular language and $L$ be a group language such that $\veps \in L$. There is a cover \Kb of $H$ such that for every language $K \in \Kb$, we have $n \in \nat$ and $a_1,\dots,a_n \in A$ such that $K = La_1L \cdots a_nL$ and $a_1 \cdots a_n \in H$.
\end{lemma}

\begin{proof}
  Since $H$ is regular, we have a morphism $\gamma: A^* \to N$ into a finite monoid $N$ recognizing $H$. Moreover, since $L$ is a group language it is recognized by a morphism $\beta: A^* \to G$ where $G$ is a finite group. Since $\veps \in L$, we have $\beta\inv(1_G) \subseteq L$.

  For every $w = a_1 \cdots a_n \in A^*$, we let $F_w \subseteq A^*$ be the language $F_w = La_1L  \cdots  a_nL$ ($F_\veps = L$). Moreover, we let $F'_w \subseteq F_w$ be the language $F'_w = \beta\inv(1_G) a_1 \beta\inv(1_G)  \cdots  a_n\beta\inv(1_G)$. Let $n = |N \times G|$. We show that the following is a cover of $H$:
  \[
    \{F'_w \mid \text{$w \in H$ and $|w| \leq n$}\}.
  \]
  Clearly, this implies that $\Kb = \{F_w \mid \text{$w \in H$ and $|w| \leq n$}\}$ is a cover of $H$, finishing the proof. We use the following fact.

  \begin{fact}\label{fct:gcovers}
    Let $w \in H$ such that $|w| > n$. Then, there exists $v \in H$ such that $|v| < |w|$ and $F'_{w} \subseteq F'_{v}$.
  \end{fact}

  We first use the fact to finish the main proof. Let $w \in H$. We have to find $K \in \{F'_w \mid \text{$w \in H$ and $|w| \leq n$}\}$ such that $w\in K$. Observe that $w \in F'_{w}$ since by definition, $\veps \in \beta\inv(1_G)$. We may then apply Fact~\ref{fct:gcovers} repeatedly to build $v \in H$ such that $|v| \leq n$ and $F'_{w} \subseteq F'_{v}$. Consequently, we have $w \in F'_{v}$ and we may choose $K = F'_v$. It remains to prove Fact~\ref{fct:gcovers}.

  We fix $w \in H$ such that $|w| > n$. We let $w = a_1 \cdots a_{|w|}$, where each $a_i$ is a letter. Clearly, $N \times G$ is a monoid for the componentwise multiplication. Let $\eta: A^* \to N \times G$ be the morphism defined by $\eta(u) = (\gamma(u),\beta(u))$ for every $u \in A^*$. Since $|w| > n = |N \times G|$, it follows from the pigeonhole principle that there exist $i,j$ such that $1 \leq i < j \leq |w|$ and $\eta(a_1 \cdots a_i) = \eta(a_1 \cdots a_j)$. Define $v = a_1 \cdots a_{i}a_{j+1} \cdots a_{|w|}$. Clearly, we have $|v| < |w|$. Moreover, $\eta(w) = \eta(v)$, hence $\gamma(w) = \gamma(v)$. Consequently, $v \in H$, since $w \in H$ and $H$ is recognized by $\gamma$. It remains to show that $F'_{w} \subseteq F'_{v}$.

  Let $u \in F'_{w}$: we have $u = u_0a_1u_1a_2u_2 \cdots a_{|w|} u_{|w|}$ with $u_i \in \beta\inv(1_G)$ for every $i \leq |w|$. Let $x = u_ia_{i+1}u_{i+1} \cdots a_{j}u_{j}$. Clearly, $\beta(x) = \beta(a_{i+1} \cdots a_j)$ and we know by definition of $i,j$ that $\beta(a_1 \cdots a_i) = \beta(a_1 \cdots a_j)$. Thus, $\beta(a_1 \cdots a_i) \cdot \beta(x) = \beta(a_1 \cdots a_i)$. Since $G$ is a group, it follows that $\beta(x) = 1_G$. Consequently, $x \in \beta\inv(1_G)$. Moreover, we have,
  \[
    \begin{array}{lll}
      u & = & u_0a_1u_1 \cdots a_i x a_{j+1}u_{j+1} \cdots a_{|w|} u_{|w|} \\
      v & = & a_1 \cdots a_{i}a_{j+1} \cdots a_{|u|}
    \end{array}
  \]
  It follows that $u \in F'_{v}$, which concludes the proof.
\end{proof}

For $t\in M$, we are now ready to construct the \pol{\Cs}-cover $\Kb_t$ of Proposition~\ref{prop:polc}. Let $L$ be a \Cs-optimal \iden for $\rho$: we have $L \in \Cs$, $\veps \in L$ and $\rho(L) = \ioptic{\rho}$. Since $\alpha\inv(t)$ is a regular language and $L$ is a group language (by hypothesis on \Cs), Lemma~\ref{lem:half:pump} yields a cover $\Kb_t$ of $\alpha\inv(t)$ such that every language $K \in \Kb_t$ is of the form $K = La_1L \cdots a_nL$ with $a_1,\dots,a_n \in A$ such that $a_1 \cdots a_n \in \alpha\inv(t)$.

Clearly, $\Kb_t$ is a \pol{\Cs}-cover of $\alpha\inv(t)$ since $L \in \Cs$. It remains to prove that for every $K \in \Kb_t$, we have $(t,\rho(K)) \in S$. By definition of $\Kb_t$, there exists $w = a_1 \cdots a_n \in \alpha\inv(t)$ such that $K = La_1L \cdots a_nL$. By definition of \pol{\Cs}-complete subsets, we have $(\alpha(a_i),\rho(a_i)) \in S$ for every $i \leq n$ (these are trivial elements). Moreover, we know from \pol{\Cs}-operation and the definition of $L$ as a \Cs-optimal \iden for~$\rho$ that $(1_M,\rho(L)) = (1_M,\ioptic{\rho}) \in S$. It then follows from closure under multiplication that:
\[
  \left(1_M \cdot\!\!\!\! \prod_{1 \leq i \leq n} (\alpha(a_i) \cdot 1_M),\quad \rho(L) \cdot\!\!\!\! \prod_{1 \leq i \leq n}(\rho(a_i) \cdot \rho(L))\right) \in S.
\]
Since $\alpha$ is a morphism and $\rho$ is a \mratm (and therefore a morphism for multiplication), this exactly says that $(\alpha(w),\rho(K)) \in S$. Finally, $\alpha(w)=t$ by definition and we obtain as desired that $(t,\rho(K)) \in S$, finishing the proof.\qed

\section{Extending the Framework}
\label{sec:units}
In this section, we introduce additional material about \ratms required to prove the remaining part of Theorem~\ref{thm:main} (\emph{i.e.}, concerning levels~1 and~3/2). Let us first overview the situation.

In Section~\ref{sec:ratms}, we proved that given an arbitrary lattice \Ds, deciding \Ds-covering boils down to computing \popti{\Ds}{\alpha}{\rho} from a morphism $\alpha$ and a \nice \mratm $\rho$ (see Proposition~\ref{prop:lreduc}). In fact, if \Ds is a Boolean algebra, we may even work with the simpler set \opti{\Ds}{\rho} (see Proposition~\ref{prop:breduc}). This is how we handled \pol{\Cs} in the previous section: we showed that given a fixed \vari of group languages \Cs, if  \Cs-separation is decidable, then \pocopti can be computed via a least fixpoint algorithm. We shall use the same approach for the classes \bpol{\Cs} and \pbpol{\Cs}: we obtain algorithms for computing \opti{\bpol{\Cs}}{\rho} and \popti{\pbpol{\Cs}}{\alpha}{\rho}. However, we do not compute these sets directly. Instead, we work with more involved sets, which \emph{carry more information}. In this section, we introduce these other sets.

\begin{remark}
  This situation is not surprising. Typically, computing optimal \imprints is achieved via fixpoint procedures (what we did for \pol{\Cs} in Section~\ref{sec:bpolg} is a typical example). Hence, replacing the object that we truly want to compute by another one which carries more information makes sense.
\end{remark}

First, we present two constructions for building new \ratms out of already existing ones. They are taken from~\cite{pzbpolc}, where they are used as technical proof objects. Here, we use them in a more prominent way: they are central to the definitions of this section. We use these constructions to present refined variants of Proposition~\ref{prop:breduc} and Proposition~\ref{prop:lreduc}.

\subsection{Nested \ratms}

We present two constructions. The first one involves two objects: a lattice \Ds and a \ratm $\rho: 2 ^{A^*} \to R$. We build a new map \bratauxd whose rating set is $2^R$. We let:
\[
  \begin{array}{llll}
    \bratauxd: & (2^{A^*},\cup) & \to     & (2^R,\cup)         \\
               & K              & \mapsto & \opti{\Ds}{K,\rho}
  \end{array}
\]
The second construction takes an additional object as input: a map $\alpha: A^* \to M$ where $M$ is some arbitrary finite set (in practice, $\alpha$ will be a monoid morphism). We build a new map \lratauxd whose rating set is $2^{M \times R}$:
\[
  \begin{array}{llll}
    \lratauxd: & (2^{A^*},\cup) & \to     & (2^{M \times R},\cup)                                      \\
               & K              & \mapsto & \{(s,r) \mid r \in \opti{\Ds}{\alpha\inv(s) \cap K,\rho}\}
  \end{array}
\]
It is straightforward to verify that these two maps are in fact \ratms (a proof is available in~\cite{pzbpolc}, see Proposition~6.2). We state this result in the following proposition.

\begin{proposition}\label{prop:areratms}
  Given a lattice \Ds, a map $\alpha: A^* \to M$ and a \ratm $\rho: 2 ^{A^*} \to R$, \bratauxd and \lratauxd are \ratms.
\end{proposition}

A key point is that the \ratms \bratauxd and \lratauxd are \textbf{not} \nice in general, even when the original \ratm $\rho$ is (we refer the reader to~\cite{pzbpolc} for a counterexample).

Another question is whether they are \tame. Clearly, when $\alpha: A^* \to M$ is a monoid morphism and $\rho: 2^{A^*} \to R$ is a \mratm, we may lift the multiplications of $M$ and $R$ to the sets $2^R$ and $2^{M \times R}$ in the natural way. This makes $(2^R,\cup,\cdot)$ and $(2^{M \times R},\cup,\cdot)$ semirings. Unfortunately, \bratauxd and \lratauxd are {\bf not} \tame, even under these hypotheses. However, they behave almost as \mratms when the class \Ds satisfies appropriate properties.

Consider an arbitrary lattice \Gs and a \ratm $\rho: 2^{A^*} \to R$ whose rating set is a semiring $(R,+,\cdot)$ (but $\rho$ need not be \tame).  We say that $\rho$ is \emph{\Gs-\tame} when there exists an endomorphism \quasir of $(R,+)$ such that:
\begin{enumerate}
\item\label{itm:ax1} For every $q,r,s \in R$, $\quasir(q\quasir(r)s) = \quasir(qrs)$.
\item\label{itm:ax2} $1_R \leq \rho(\veps)$.
\item\label{itm:ax3} For every $H,K \in \Gs$ and $a \in A$, we have,
  \[
    \begin{array}{lll}
      \rho(H)     & = & \quasir(\rho(H))                             \\
      \rho(HK)    & = & \quasir\big(\rho(H) \cdot \rho(K)\big)               \\
      \rho(HaK)   & = & \quasir\big(\rho(H) \cdot \rho(a) \cdot \rho(K)\big)
    \end{array}
  \]
\end{enumerate}
When \Gs is the class of all languages (\emph{i.e.}, $\Gs = 2^{A^*}$), we say that $\rho$ is \emph{quasi-\tame}. In practice, we shall always assume implicitly that the endomorphism \quasir is fixed.

\begin{remark}\label{rem:truem}
  A true \mratm is always quasi-\tame. Indeed, in this case, it suffices to choose $\quasir$ as the identity; $\quasir(r) = r$ for all $r \in R$.
\end{remark}

We have the two following lemmas, which apply to \lratauxd and \bratauxd respectively (see~\cite[Lemmas~6.7 and~6.8]{pzbpolc} for proofs). The statements are rather \emph{ad hoc}: they are designed to accommodate the situations that we shall encounter. We use \lratauxd in the case when $\Ds = \pol{\Cs}$ while we use \bratauxd in the case when $\Ds = \bpol{\Cs}$.

\begin{lemma}\label{lem:copelrat}
  Let \Ds be a \pvari closed under concatenation, $\alpha: A^* \to M$ be a morphism and $\rho: 2 ^{A^*} \to~R$ be a \mratm. Then, \lratauxd is quasi-\tame, and the endomorphism $\quasi{\lratauxd}$ of $(2^{M \times R},\cup)$~is:
  \[
    \quasi{\lratauxd}(T) = \dclosr T \quad \text{for all $T \in 2^{M \times R}$}.
  \] \end{lemma}

\begin{lemma}\label{lem:copelbrat}
  Let \Gs be a \pvari closed under concatenation and {\bf marked} concatenation and let \Ds = \bool{\Gs}. Let $\rho: 2 ^{A^*} \to R$ a \mratm. Then, \bratauxd is \Gs-\tame and the endomorphism $\quasi{\bratauxd}$ of $(2^R,\cup)$ is:
  \[
    \quasi{\bratauxd}(T) = \dclosr T \quad \text{for all $T \in 2^{R}$}.
  \]
\end{lemma}

\subsection{Application to classes of group languages}

We may now present the refinements of Proposition~\ref{prop:breduc} and Proposition~\ref{prop:lreduc} announced at the beginning. They are designed to handle the classes investigated in the paper. Here, we discuss and prove the variant for Boolean algebras and simply state the one for lattices (we prove it in appendix).

By Proposition~\ref{prop:breduc}, given a Boolean algebra \Ds, deciding \Ds-covering boils down to computing $\opti{\Ds}{\rho} \subseteq R$ from a \nice \mratm $\rho: 2^{A^*} \to R$. It turns out that when \Ds is a full level in some concatenation hierarchy whose basis \Cs is made of group languages, one may instead compute a specific subset of \opti{\Ds}{\rho}.

By definition, $\opti{\Ds}{\rho} = \opti{\Ds}{A^*,\rho} = \bratauxd(A^*)$. We may consider the element $\ioptic{\bratauxd} \in 2^R$ which is equal to $\bratauxd(L)$ for a \Cs-optimal \iden $L$ for $\bratauxd$. Hence, we have,
\[
  \ioptic{\bratauxd} = \bratauxd(L) \subseteq \bratauxd(A^*) = \opti{\Ds}{\rho}.
\]
It turns out that because of our hypothesis on \Ds, one may compute \opti{\Ds}{\rho} from its subset \ioptic{\bratauxd} (see below for the proof). This yields the following corollary of Proposition~\ref{prop:breduc}.

\begin{proposition}\label{prop:utooptibool}
  Let \Cs be a \vari of group languages and let \Ds be a strictly positive full level within the concatenation hierarchy of basis \Cs. Assume that there exists an algorithm computing \ioptic{\bratauxd} from an input \nice \mratm $\rho$. Then, \Ds-covering is decidable.
\end{proposition}

Before we prove Proposition~\ref{prop:utooptibool}, let us present the generalized variant of Proposition~\ref{prop:lreduc} which considers all lattices instead of just Boolean algebras (we detail it in appendix).

\begin{proposition}\label{prop:utooptilat}
  Let \Cs be a \vari of group languages and let \Ds be a half level within the concatenation hierarchy of basis \Cs. Assume that there exists an algorithm computing \ioptic{\lratauxd} from a morphism $\alpha$ and a \nice \mratm $\rho$. Then, \Ds-covering is decidable.
\end{proposition}

\subsection{Proof of Proposition~\ref{prop:utooptibool}}

We fix \Cs and \Ds as in the statement of the proposition. We show that for an arbitrary \nice \mratm $\rho$, one may compute \opti{\Ds}{\rho} from its subset \ioptic{\bratauxd}. The result is then an immediate corollary of Proposition~\ref{prop:breduc}.

Fix a \nice \mratm $\rho: 2^{A^*} \to R$ for the proof. We let $S \subseteq R$ be the least subset of $R$ that satisfies the following properties:
\begin{enumerate}
\item We have $\ioptic{\bratauxd} \subseteq S$.
\item We have $\rho(w) \in S$ for every $w \in A^*$.
\item We have $\dclosr S = S$.
\item For every $r,r' \in S$, we have $rr' \in S$.
\end{enumerate}
Clearly, one may compute $S$ from $\rho$ and \ioptic{\bratauxd}. Hence, Proposition~\ref{prop:utooptibool} is now immediate from the following lemma.

\begin{lemma}\label{lem:utooptibool}
  We have $S = \opti{\Ds}{\rho}$.
\end{lemma}

We prove Lemma~\ref{lem:utooptibool}. Let $L$ be a \Cs-optimal \iden for \bratauxd: we have $L \in \Cs$, $\veps \in L$ and $\ioptic{\bratauxd} = \bratauxd(L)$.

We first show that $S \subseteq \opti{\Ds}{\rho}$. This amounts to proving that \opti{\Ds}{\rho} satisfies the four properties in the definition of $S$. That $\ioptic{\bratauxd} \subseteq \opti{\Ds}{\rho}$ is immediate since $\ioptic{\bratauxd} = \bratauxd(L)$, $\opti{\Ds}{\rho} = \bratauxd(A^*)$, $L \subseteq A^*$ and \bratauxd is a \ratm. The other three properties are generic: they are satisfied whenever $\rho$ is a \mratm and \Ds is a \vari (which is the case here as \Ds is a full level within the hierarchy of basis \Cs). We refer the reader to~\cite[Lemma~6.11]{pzcovering2} for the proof.

We turn to the converse inclusion: $\opti{\Ds}{\rho} \subseteq S$. Consider $r \in \opti{\Ds}{\rho}$, we show that $r \in S$. The argument is based on the following lemma, which is where we use the hypothesis that \Cs is a class of group languages (which implies that $L \in \Cs$ is recognized by a finite group).

\begin{lemma}\label{lem:grouplem}
  There exist $\ell \in \nat$ and $\ell$ letters $a_1,\dots,a_\ell \in A$ such that $r \in \bratauxd(La_1L \cdots a_\ell L)$.
\end{lemma}

\begin{proof}
  For every word $w = a_1 \cdots a_\ell \in A^*$, we let $H_w$ be the language $H_w = La_1L \cdots a_\ell L$ ($H_\veps = L$). We have to find $w \in A^*$ such that $r \in \bratauxd(H_w)$. Clearly $A^*$ is a regular language and $L$ is a group language by hypothesis. Therefore, Lemma~\ref{lem:half:pump} implies that there exists a \emph{finite} language $U$ such that $A^* \subseteq \bigcup_{w \in U} H_w$. Since \bratauxd is a \ratm, this implies:
  \[
    \bratauxd(A^*) = \bigcup_{w \in U}  \bratauxd(H_w).
  \]
  Moreover, by hypothesis, we have $r \in \opti{\Ds}{\rho} = \opti{\Ds}{A^*,\rho} = \bratauxd(A^*)$. Thus, $r \in \bratauxd(H_w)$ for some $w \in A^*$ and the lemma is proved.
\end{proof}

By definition \Ds is a \textbf{strictly positive} full level within the hierarchy of basis \Cs. Therefore, $\Ds = \bool{\Gs}$ where \Gs is the preceding half level in the hierarchy. In particular, \Gs is a \pvari closed under concatenation and marked concatenation by Lemma~\ref{lem:boolclos} and Theorem~\ref{thm:polclos}. Consequently, we obtain from Lemma~\ref{lem:copelbrat} that \bratauxd is \Gs-\tame and the associated endomorphism of $(2^R,\cup)$ is defined by $\quasi{\bratauxd}(T) = \dclosr T$ for all $T \in 2^{R}$. Moreover, we have $L \in \Cs \subseteq \Gs$ by definition. Hence, by Axioms~\ref{itm:ax1} and~\ref{itm:ax3} in the definition of \Gs-\mratms, the hypothesis that $r \in \bratauxd(La_1L \cdots a_\ell L)$ given by Lemma~\ref{lem:grouplem} implies that,
\[
  r \in \dclosr(\bratauxd(L) \cdot \bratauxd(a_1) \cdot \bratauxd(L) \cdots \bratauxd(a_\ell) \cdot \bratauxd(L))
\]
By definition of $L$, we have $\bratauxd(L) = \ioptic{\bratauxd}$. Moreover, we have the following fact.

\begin{fact}\label{fct:theletters}
  For every $a \in A$, we have $\bratauxd(a) \subseteq \dclosr \{\rho(a)\}$.
\end{fact}

\begin{proof}
  Fact~\ref{fct:finitelangs} yields that $\{a\} \in \Ds$ (\Ds is a level $n \geq 1$ in the hierarchy of basis~\Cs). Hence, $\{\{a\}\}$ is a \Ds-cover of $\{a\}$. Thus, $\bratauxd(a) = \opti{\Ds}{\{a\},\rho} \subseteq \prin{\rho}{\{a\}} = \dclosr \{\rho(a)\}$, which finishes the proof.
\end{proof}

Altogether, we obtain:
\[
  r \in \dclosr(\ioptic{\bratauxd} \cdot \{\rho(a_1)\} \cdot \ioptic{\bratauxd} \cdots \{\rho(a_\ell)\} \cdot \ioptic{\bratauxd}).
\]
By definition of $S$, this implies $r \in S$, concluding the proof.

\section{Boolean polynomial closure}
\label{sec:bpolg}
We turn to the second part of Theorem~\ref{thm:main}: for a \vari of group languages \Cs, if \Cs-separation~is decidable, then so is \bpol{\Cs}-covering. We fix \Cs for the section.

We rely on \ratms and use the notions introduced in Section~\ref{sec:units}. For any \emph{\nice} \mratm $\rho: 2^{A^*} \to R$, we  characterize \ioptic{\bratauxbc} as the greatest subset of $R$ satisfying specific properties. When \Cs-separation is decidable, this yields a fixpoint algorithm for computing \ioptic{\bratauxbc} from $\rho$. Consequently, we get the decidability of \bpol{\Cs}-covering from Proposition~\ref{prop:utooptibool}.

\begin{remark}\label{rem:bpol}
  The conditions for applying our characterizations are more restrictive than those we had for \pol{\Cs} in the previous section.  We require $\rho$ to be \nice: while we are able to handle \pol{\Cs} for arbitrary \mratms, we are restricted to \nice ones for \bpol{\Cs}. This is irrelevant for the decidability of covering: considering \nice \mratms suffices. However, this is a key point: the proof of the \bpol{\Cs}-characterization involves handling the simpler class \pol{\Cs} for auxiliary \pol{\Cs}-\mratms which are \textbf{not} \nice. This is why we are not able to get results for higher levels in concatenation hierarchies: our knowledge about level 1/2 is stronger than the decidability of covering, and we are unable to replicate it for levels 1 and 3/2 (the situation is the same for \pbpol{\Cs}).
\end{remark}

\newcommand{\bwmrat}[1]{\ensuremath{\eta_{\rho,#1}}\xspace}
\newcommand{\bwmrats}{\bwmrat{S}}

We first present the characterization of \bpoluopti. For every \mratm $\rho: 2^{A^*} \to R$, we shall define the \emph{\bpol{\Cs}-complete subsets of $R$ for $\rho$}. Our theorem states that when $\rho$ is \nice, \bpoluopti is the \emph{greatest} such subset.

We fix the \mratm $\rho: 2^{A^*} \to R$ for the definition. We successively lift the multiplication ``$\cdot$'' of $R$ to the sets $2^R$, $R \times 2^R$ and $2^{R \times 2^R}$ in the natural way. This makes $(2^{R \times 2^R},\cup,\cdot)$ an idempotent semiring. For every  $S \subseteq R$, we use it as the rating set of a \nice \mratm,
\[
  \bwmrats: (2^{A^*},\cup,\cdot) \to (2^{R \times 2^R},\cup,\cdot).
\]
Since we are defining a \nice \mratm, it suffices to specify the evaluation of letters. For $a \in A$, we let,
\[
  \bwmrats(a)  =  \big\{(\rho(a), \quad S \cdot \{\rho(a)\} \cdot S)\big\}.
\]
By definition, we have $\ioptic{\bwmrats} \subseteq R \times 2^R$ for every $S \subseteq R$.

We now define the \bpol{\Cs}-complete subsets. We say that $S \subseteq R$ is \bpol{\Cs}-complete for $\rho$ if for every $s \in S$, there exist $(r_1,U_1),\dots,(r_k,U_k) \in \ioptic{\bwmrats}$ such that,
\begin{equation}\label{eq:bpolcarac}
  \begin{array}{c}
    \text{$s \leq r_1 + \cdots + r_k$ and,}\\ \text{$r_1 + \cdots + r_k \in\ \dclosr U_i$ for every $i \leq k$}
  \end{array}
\end{equation}

We are ready to state the main theorem of this section: when $\rho$ is \nice, the greatest \bpol{\Cs}-complete subset of $R$ (with respect to inclusion) is exactly \bpoluopti.

\begin{theorem}\label{thm:bpolg}
  Let $\rho: 2^{A^*} \to R$ be a \nice \mratm. Then, \bpoluopti is the greatest \bpol{\Cs}-complete subset of~$R$ for $\rho$.
\end{theorem}

The proof of Theorem~\ref{thm:bpolg} is presented in Appendix~\ref{app:bpolg}. It is rather involved and exploits two results for \ratms that were designed in~\cite{pzbpolc} to handle Boolean closure in a general context. We complete them with arguments which are specific to our setting. Here, we discuss the applications of Theorem~\ref{thm:bpolg}.

Provided that \Cs-separation is decidable, Theorem~\ref{thm:bpolg} yields a greatest fixpoint procedure for computing \bpoluopti from an input \nice \mratm $\rho: 2^{A^*} \to R$. Indeed, consider the following sequence of subsets $S_0 \supseteq S_1 \supseteq S_2 \cdots$. We let $S_0 = R$ and for $i \geq 1$, $S_i$ contains all $s \in S_{i-1}$ such that there exists $(r_1,U_1),\dots,(r_k,U_k) \in \ioptic{\bwmrat{S_{i-1}}}$ satisfying~\eqref{eq:bpolcarac}:
\[
  \begin{array}{c}
    \text{$s \leq r_1 + \cdots + r_k$ and,}\\ \text{$r_1 + \cdots + r_k \in\ \dclosr U_i$ for all $i \leq k$}
  \end{array}
\]
Clearly, computing $S_i$ from $S_{i-1}$ boils down to computing \ioptic{\bwmrat{S_{i-1}}}. By Corollary~\ref{cor:epswit2}, this is possible since we have an algorithm for \Cs-separation (\bwmrat{S_{i-1}} is a \textbf{\nice} \mratm that we may compute from $S_{i-1}$).

Finally, since $R$ is finite, the sequence $S_0 \supseteq S_1 \supseteq S_2 \cdots$ stabilizes at some point: there exists some $i \in \nat$ such that $S_i = S_{j}$ for all $j \geq i$. One may verify that $S_i$ is the greatest \bpol{\Cs}-complete subset of $R$ and we may compute it. Theorem~\ref{thm:bpolg} then states that $S_i = \bpoluopti$.

By Proposition~\ref{prop:utooptibool}, a procedure for computing \bpoluopti yields an algorithm for \bpol{\Cs}-covering. Hence, we get the part of Theorem~\ref{thm:main} regarding \bpol{\Cs}  (\emph{i.e.}, level~1) as a corollary: when \Cs-separation is decidable, so is \bpol{\Cs}-covering (as well as \bpol{\Cs}-separation by Lemma~\ref{lem:septocove}).

\section{Nested polynomial closure}
\label{sec:pbpolg}
This section is devoted to the final part of Theorem~\ref{thm:main}: for every \vari of group languages \Cs, if \Cs-separation is decidable, then so is  \pbpol{\Cs}-covering. As usual, we fix \Cs for the section.

We rely on the framework outlined in Sections~\ref{sec:ratms} and~\ref{sec:units}. For every morphism $\alpha: A^* \to M$ and every \emph{\nice} \mratm $\rho: 2^{A^*} \to R$, we characterize \ioptic{\lratauxppc} as the least subset of $M \times R$ satisfying specific properties. When \Cs-separation is decidable, this yields a least fixpoint algorithm for computing \ioptic{\lratauxppc} from~$\rho$. Consequently, we get the decidability of \pbpol{\Cs}-covering by Proposition~\ref{prop:utooptilat}.

\newcommand{\pbwmrat}[1]{\ensuremath{\eta_{\alpha,\rho,#1}}\xspace}
\newcommand{\pbwmrats}{\pbwmrat{S}}

\medskip

We start by presenting the characterization. Consider a morphism $\alpha: A^* \to M$ and a \mratm $\rho: 2^{A^*} \to R$. We define a notion of \pbpol{\Cs}-complete subset of $M \times R$. Our main theorem then states that when $\rho$ is \nice, the least such subset is exactly \pbpoluopti. The definition depends on auxiliary \nice \mratms. We first present them.

We successively lift the multiplications of $M$ and $R$ to the sets $2^{M \times R}$, $R \times 2^{M \times R}$ and $2^{R \times 2^{M \times R}}$ in the natural way. This makes $(2^{R \times 2^{M \times R}},\cup,\cdot)$ an idempotent semiring. For every  $S \subseteq M \times R$, we define a \nice \mratm:
\[
  \pbwmrats: (2^{A^*},\cup,\cdot) \to (2^{R \times 2^{M \times R}},\cup,\cdot).
\]
Since we are defining a \nice \mratm, it suffices to specify the evaluation of letters. For $a \in A$, we let,
\[
  \pbwmrats(a)  =  \big\{(\rho(a),\quad S \cdot \{(\alpha(a),\rho(a))\} \cdot S)\big\}
\]
Observe that by definition, we have $\ioptic{\pbwmrats} \subseteq R \times 2^{M \times R}$.

We may now define \pbpol{\Cs}-complete subsets. Consider $S \subseteq M \times R$. We say that $S$ is \pbpol{\Cs}-complete for $\alpha$ and~$\rho$ when the following conditions are satisfied:
\begin{itemize}
\item \emph{\bfseries Downset.} We have $\dclosr S \subseteq S$.
\item \emph{\bfseries Multiplication.} We have $S \cdot S \subseteq S$.
\item \emph{\bfseries \Cs-operation.} For all $(r,T) \in \ioptic{\pbwmrats}$, we have $T \subseteq S$.
\item \emph{\bfseries \pbpol{\Cs}-operation.} For all $(r,T) \in \ioptic{\pbwmrats}$ and every idempotent $(e,f) \in \dclosr T \subseteq M \times R$, we~have:
  \[
    (e, f \cdot (1_R + r) \cdot f) \in S.
  \]
\end{itemize}

We may now state the main theorem of this section. When $\rho$ is \nice, \pbpoluopti is the least \pbpol{\Cs}-complete subset of $M \times R$ (with respect to inclusion).

\begin{theorem}\label{thm:pbpolg}
  Fix a morphism $\alpha: A^* \to M$ and a \nice \mratm $\rho: 2^{A^*} \to R$. Then, \pbpoluopti is the least \pbpol{\Cs}-complete subset of $M \times R$.
\end{theorem}

The proof of Theorem~\ref{thm:pbpolg} is presented in Appendix~\ref{app:pbpolg}. Similarly to most results of this kind, it involves two directions: one needs to show that \pbpoluopti is \pbpol{\Cs}-complete and then that it is included in every other \pbpol{\Cs}-complete subset. The former direction is proved directly. The latter one requires applying the main theorem of~\cite{pseps3j} (for classes \pbpol{\Ds} with \Ds a {\bf finite} \vari).

As expected, when \Cs-separation is decidable, Theorem~\ref{thm:pbpolg} yields a least fixpoint procedure for computing \pbpoluopti from a morphism $\alpha: A^* \to M$ and a \nice \mratm $\rho: 2^{A^*} \to R$. One starts from the empty set $\emptyset \subseteq R$ and saturates it with the four operations in the definition of \pbpol{\Cs}-complete subsets. Clearly, they may be implemented. This is immediate for downset and multiplication. Moreover, we are able to implement \Cs-operation and \pbpol{\Cs}-operation by Corollary~\ref{cor:epswit2} since \Cs-separation is decidable. Eventually, the computation reaches a fixpoint and it is straightforward to verify that this set is the least  \pbpol{\Cs}-complete subset of $M \times R$, \emph{i.e.}, \pbpoluopti by Theorem~\ref{thm:pbpolg}.

By Proposition~\ref{prop:utooptilat}, we know that a procedure for computing \pbpoluopti yields an algorithm for \pbpol{\Cs}-covering. Hence, we get the \pbpol{\Cs} part of Theorem~\ref{thm:main} as a corollary: when \Cs-separation is decidable, so is \pbpol{\Cs}-covering (as well as \pbpol{\Cs}-separation by Lemma~\ref{lem:septocove}).

\section{Conclusion}
\label{sec:conc}
We proved that for any \vari of \emph{group languages} \Cs, if \Cs-separation is decidable, then so are separation and covering for levels 1/2,  1 and 3/2 in the concatenation hierarchy of basis \Cs. A corollary is that these levels enjoy decidable membership, as well as level 5/2. This result nicely complements analogous statements that apply to \emph{finitely based} concatenation hierarchies~\cite{pzbpolc,pseps3j}.

These results can be instantiated for several classical bases. First, one may consider the basis made of all group languages. This yields the so-called ``group hierarchy'' of Pin and Margolis~\cite{MargolisP85}, for which decidability of membership was known only up to level~1. Another application is the quantifier alternation hierarchy of first-order logic with modular predicates. It corresponds to the basis consisting of languages counting the length of words modulo some number (separation was shown decidable at level~1 in~\cite{Zetzsche18} with specific techniques). A third example is the basis consisting of all languages counting the number of occurrences of letters modulo some number, which are exactly languages recognized by finite commutative groups. In all cases, we get decidability of separation and covering for levels 1/2, 1 and 3/2 (and membership for level~5/2).

There are natural follow-up questions to this work. The most immediate one is whether our results may be pushed to higher levels. This is difficult. There are no known generic results for the levels above 2, even for finitely based hierarchies. Actually, there is one hierarchy for which covering is know to be decidable up to level 5/2: the Straubing-Thérien hierarchy (its basis is $\{\emptyset,A^*\}$). Unfortunately, this is based on a specific property of this hierarchy: its levels 2 and 5/2 are also levels 1 and 3/2 in another finitely based hierarchy~\cite{pin-straubing:upper}. We do not have a similar property for arbitrary group based hierarchies.

One can also investigate other closure operations, such as the star-free closure $\Cs \mapsto \sfr(\Cs)$, which is the union of all levels in the hierarchy of basis \Cs (\emph{i.e.}, $\sfr(\Cs)$ is the least class containing \Cs, closed under Boolean operations and marked concatenation). If \Cs is a \hbox{\vari}~of group languages and \Cs-separation is decidable, is $\sfr(\Cs)$-covering decidable?

% \bibliographystyle{abbrv}
% \bibliography{main}

\begin{thebibliography}{10}

\bibitem{MR1709911}
J.~Almeida.
\newblock Some algorithmic problems for pseudovarieties.
\newblock {\em Publicationes Mathematicae Debrecen}, 54:531--552, 1999.

\bibitem{arfi87}
M.~Arfi.
\newblock Polynomial operations on rational languages.
\newblock In {\em STACS'87}, Lect. Notes Comp. Sci., pages 198--206, 1987.

\bibitem{Ash91}
C.~J. Ash.
\newblock Inevitable graphs: a proof of the type {II} conjecture and some
  related decision procedures.
\newblock {\em {IJAC}}, 1(1):127--146, 1991.

\bibitem{BrzoDot}
J.~A. Brzozowski and R.~S. Cohen.
\newblock Dot-depth of star-free events.
\newblock {\em J. Comp. Sys. Sci.}, 5(1):1--16, 1971.

\bibitem{BroKnaStrict}
J.~A. Brzozowski and R.~Knast.
\newblock The dot-depth hierarchy of star-free languages is infinite.
\newblock {\em J. Comp. Sys. Sci.}, 16(1):37--55, 1978.

\bibitem{ChaubardPS06}
L.~Chaubard, J.~Pin, and H.~Straubing.
\newblock First order formulas with modular predicates.
\newblock In {\em LICS'06}, 2006.

\bibitem{abelian_pt}
M.~Delgado.
\newblock Abelian pointlikes of a monoid.
\newblock {\em Semigroup Forum}, 56:339--361, 1998.

\bibitem{Eilenberg_book_B}
S.~Eilenberg.
\newblock {\em Automata, languages, and machines}, volume~B.
\newblock Academic Press, 1976.

\bibitem{gssig2}
C.~Gla{\ss}er and H.~Schmitz.
\newblock Languages of dot-depth 3/2.
\newblock In {\em STACS'00}, pages 555--566. Springer, 2000.

\bibitem{HenckellMPR91}
K.~Henckell, S.~W. Margolis, J.~Pin, and J.~Rhodes.
\newblock Ash's type {II} theorem, profinite topology and malcev products: Part
  {I}.
\newblock {\em {IJAC}}, 1(4):411--436, 1991.

\bibitem{knast83}
R.~Knast.
\newblock A semigroup characterization of dot-depth one languages.
\newblock {\em RAIRO TIA}, 17(4):321--330, 1983.

\bibitem{KufleitnerW15}
M.~Kufleitner and T.~Walter.
\newblock One quantifier alternation in first-order logic with modular
  predicates.
\newblock {\em {RAIRO} TIA}, 49(1), 2015.

\bibitem{MargolisP85}
S.~W. Margolis and J.~Pin.
\newblock Products of group languages.
\newblock In {\em {FCT}'85}. Springer, 1985.

\bibitem{mnpfo}
R.~McNaughton and S.~A. Papert.
\newblock {\em Counter-Free Automata}.
\newblock {MIT} Press, 1971.

\bibitem{pinbridges}
J.-{\'E}. Pin.
\newblock Bridges for concatenation hierarchies.
\newblock In {\em ICALP'98}, pages 431--442. Springer, 1998.

\bibitem{jep-intersectPOL}
J.-{\'E}. Pin.
\newblock An explicit formula for the intersection of two polynomials of
  regular languages.
\newblock In {\em {DLT 2013}}, pages 31--45. Springer, 2013.

\bibitem{jep-dd45}
J.-{\'E}. Pin.
\newblock The dot-depth hierarchy, 45 years later.
\newblock In {\em The Role of Theory in Computer Science. Essays Dedicated to
  Janusz Brzozowski}. 2017.

\bibitem{pingoodref}
J.-{\'E}. Pin.
\newblock Mathematical foundations of automata theory.
\newblock 2018.

\bibitem{pin-straubing:upper}
J.-{\'E}. Pin and H.~Straubing.
\newblock Monoids of upper triangular boolean matrices.
\newblock In {\em Semigroups. Structure and Universal Algebraic Problems},
  volume~39, pages 259--272. North-Holland, 1985.

\bibitem{pseps3}
T.~Place.
\newblock Separating regular languages with two quantifiers alternations.
\newblock In {\em LICS'15}, pages 202--213. {IEEE} Computer Society, 2015.

\bibitem{pseps3j}
T.~Place.
\newblock Separating regular languages with two quantifier alternations.
\newblock {\em Logical Methods in Computer Science}, 14(4), 2018.

\bibitem{pzqalt}
T.~Place and M.~Zeitoun.
\newblock Going higher in the first-order quantifier alternation hierarchy on
  words.
\newblock In {\em ICALP'14}, pages 342--353, 2014.

\bibitem{pzsucc}
T.~Place and M.~Zeitoun.
\newblock Separation and the successor relation.
\newblock In {\em STACS'15}, pages 662--675. Springer, 2015.

\bibitem{PZ:Siglog15}
T.~Place and M.~Zeitoun.
\newblock The tale of the quantifier alternation hierarchy of first-order logic
  over words.
\newblock {\em SIGLOG news}, 2(3):4--17, 2015.

\bibitem{pzboolpol}
T.~Place and M.~Zeitoun.
\newblock Separation for dot-depth two.
\newblock In {\em {32th Annual ACM/IEEE Symposium on Logic in Computer
  Science}}, LICS'17, 2017.

\bibitem{pzcovering2}
T.~Place and M.~Zeitoun.
\newblock The covering problem.
\newblock {\em Logical Methods in Computer Science}, 14(3), 2018.

\bibitem{pzgenconcat}
T.~Place and M.~Zeitoun.
\newblock Generic results for concatenation hierarchies.
\newblock {\em Theory of Computing Systems (ToCS)}, 2018.
\newblock Selected papers from CSR'17.

\bibitem{pzbpolc}
T.~Place and M.~Zeitoun.
\newblock Separation for dot-depth two.
\newblock In preparation, preprint
  at~\url{http://www.labri.fr/~zeitoun/research/pdf/boolpol-full.pdf}, 2018.

\bibitem{sfo}
M.~P. Sch{\"u}tzenberger.
\newblock On finite monoids having only trivial subgroups.
\newblock {\em Information and Control}, 8(2):190--194, 1965.

\bibitem{StrauConcat}
H.~Straubing.
\newblock A generalization of the {Sch{\"u}tzenberger} product of finite
  monoids.
\newblock {\em Theoret. Comp. Sci.}, 13(2):137--150, 1981.

\bibitem{StrauVD}
H.~Straubing.
\newblock Finite semigroup varieties of the form {V {\textasteriskcentered} D}.
\newblock {\em Journal of Pure and Applied Algebra}, 1985.

\bibitem{TheConcat}
D.~Th{\'e}rien.
\newblock Classification of finite monoids: The language approach.
\newblock {\em Theoretical Computer Science}, 14(2):195--208, 1981.

\bibitem{ThomEqu}
W.~Thomas.
\newblock Classifying regular events in symbolic logic.
\newblock {\em J. Comp. Sys. Sci.}, 25(3):360--376, 1982.

\bibitem{Zetzsche18}
G.~Zetzsche.
\newblock Separability by piecewise testable languages and downward closures
  beyond subwords.
\newblock In {\em {LICS}'18}, 2018.

\end{thebibliography}

\newpage

\section{Omitted proofs in Section~\ref{sec:polg}}
\label{app:polg}
In this appendix, we present the arguments for the statements of Section~\ref{sec:units} whose proofs are missing.

\subsection{Proof of Lemma~\ref{lem:epswit}}

Let us first recall the statement of Lemma~\ref{lem:epswit}.

\adjustc{lem:epswit}
\begin{lemma}
  Let $\rho: 2^{A^*} \to R$ be a \ratm and \Cs a lattice. There exists a \Cs-optimal \iden for $\rho$.
\end{lemma}
\restorec
\begin{proof}
  Let $U \subseteq R$ be the set of all elements $r \in R$ such that $r = \rho(K)$ for $K \in \Cs$ such that $\veps \in K$. Clearly, $\rho(A^*) \in U$ which means that $U$ is non-empty since $A^* \in \Cs$ (\Cs is a lattice). For all $r \in U$, we fix an arbitrary language $K_r \in \Cs$ such that $\veps \in K_r$ and $r = \rho(K_r)$ ($K_r$ exists by definition of $U$). Finally, we let,
  \[
    K = \bigcap_{r \in U} K_r
  \]
  Since \Cs is a lattice, we have $K \in \Cs$. Moreover, $\veps \in K$ by definition. Since $K \subseteq K_r$ for all $r \in U$, it follows that $\rho(K) \leq r$ for every $r \in U$. By definition of $U$, this implies that $\rho(K) \leq \rho(K')$ for every $K' \in \Cs$ such that $\veps \in K'$. Hence, $K$ is a \Cs-optimal \iden for $\rho$.
\end{proof}

\subsection{Proof of Lemma~\ref{lem:sepepswit}}

Let us first recall the statement of Lemma~\ref{lem:sepepswit}.

\adjustc{lem:sepepswit}
\begin{lemma}
  Let $\rho: 2^{A^*} \to R$ be a \nice \ratm and \Cs a lattice. Then, $\iopti{\Cs}{\rho}$ is the sum of all $r \in R$ such that $\{\veps\}$ is not \Cs-separable from $\rho_*\inv(r)$.
\end{lemma}
\restorec

\begin{proof}
  We let $Q \subseteq R$ as the set of all $r \in R$ such that $\{\veps\}$ is not \Cs-separable from $\rho_*\inv(r)$. Moreover, we let $q = \sum_{r \in Q} r$. We prove that $q = \iopti{\Cs}{\rho}$.  First, we show that $q \leq \iopti{\Cs}{\rho}$. This amounts to proving that given $r \in Q$, we have $r \leq \iopti{\Cs}{\rho}$. By definition, $\iopti{\Cs}{\rho} = \rho(K)$ for some $K \in \Cs$ such that $\veps \in K$. Therefore, since $\{\veps\}$ is not \Cs-separable from $\rho_*\inv(r)$ (this is the definition of $Q$), we obtain that $K \cap\rho_*\inv(r) \neq \emptyset$. We get a word $w \in K \cap\rho_*\inv(r)$. By definition $w \in K$ and $\rho(w) = r$. This implies that $r \leq \rho(K) = \iopti{\Cs}{\rho}$.

  Conversely, we show that $\iopti{\Cs}{\rho} \leq q$. By definition, for  every $r \in R \setminus Q$, $\{\veps\}$ is \Cs-separable from $\rho_*\inv(r)$. Hence, we may fix $H_r \in \Cs$ as a separator: $\veps \in H_r$ and $H_r \cap \rho_*\inv(r) = \emptyset$. We define,
  \[
    H = \bigcap_{r \in R \setminus Q} H_r \in \Cs
  \]
  Clearly, $\veps \in H$ which implies that $\iopti{\Cs}{\rho} \leq \rho(H)$. Moreover, since $\rho$ is \nice, we have $w_1,\dots,w_n \in H$ such that $\rho(H) = \rho(w_1) + \cdots + \rho(w_n)$. Finally, it is clear that $H \cap \rho_*\inv(r) = \emptyset$ for every $r \in R \setminus Q$. Therefore, that $\rho(w_1),\dots,\rho(w_n) \in Q$. This yields $\rho(H) = \rho(w_1) + \cdots + \rho(w_n) \leq \sum_{r \in Q} r = q$. Altogether, we get that $\iopti{\Cs}{\rho} \leq q$.
\end{proof}

\subsection{Proof of Lemma~\ref{lem:mdiopti}}

Let us first recall the statement of Lemma~\ref{lem:mdiopti}.

\adjustc{lem:mdiopti}
\begin{lemma}
  Let $\rho: 2^{A^*} \to R$ be a \nice \mratm. Then, $\iopti{\md}{\rho} = (\rho(A))^\omega + \rho(\veps)$.
\end{lemma}
\restorec

\begin{proof}
  We first prove that $(\rho(A))^\omega + \rho(\veps) \leq \iopti{\md}{\rho}$. Let $K$ be \md-optimal \iden for $\rho$: we have $K \in \md$, $\veps \in K$ and $\iopti{\md}{\rho} = \rho(K)$. Therefore, we have to show that $(\rho(A))^\omega + \rho(\veps) \leq \rho(K)$. By definition of \md, $K$ is a Boolean combination of languages $\{w \in A^* \mid |w| = k \mod m\}$ with $k,m \in \nat$ such that $k < m$. We let $d$ as the least common multiplier of the numbers $m$ involved in the definitions of all languages in this Boolean combination. It is simple to verify that given two words $w,w' \in A^*$, if the lengths of $w$ and $w'$ are congruent modulo $d$, then $w \in K$ if and only if $w' \in K$. Hence, since $\veps \in K$. We know that every $w \in A^*$ whose length is congruent to $0$ modulo $d$ belongs to $K$. In particular, we have,
  \[
    A^{d\omega} \cup \{\veps\} \subseteq K
  \]
  This implies as desired that $(\rho(A))^\omega + \rho(\veps) \leq \rho(K)$ since $\rho$ is a \mratm.

  We now prove that $\iopti{\md}{\rho} \leq (\rho(A))^\omega + \rho(\veps)$. By definition, this amounts to finding $K \in \md$ such that $\veps \in K$ and $\rho(K) \leq (\rho(A))^\omega + \rho(\veps)$. Recall that $\omega \in \nat$ is the idempotent power of $(R,\cdot)$. Hence, we may define $K$ as the language of all words $w \in A^*$ whose lengths are congruent to $0$ modulo $\omega$. Clearly, $K \in \md$ and $\veps \in K$. Since $\rho$ is \nice, we get $w_1,\dots,w_n \in K$ such that $\rho(K) = \rho(w_1) + \cdots + \rho(w_n)$. Let $i \leq n$, by definition of $K$, one of the two following properties holds:
  \begin{itemize}
  \item $w_i = \veps$ which implies that $\rho(w_i) = \rho(\veps)$, or,
  \item $w_i \in A^+$ and $|w_i| = k\omega$ for some $k \geq 1$ which implies that $\rho(w_i) \leq \rho(A^{k\omega}) = (\rho(A))^\omega$.
  \end{itemize}
  Altogether, this yields that $\rho(K) \leq (\rho(A))^\omega + \rho(\veps)$, finishing the proof of this direction.
\end{proof}

\section{Omitted proofs in Section~\ref{sec:units}}
\label{app:units}
In this appendix, we present an additional property of quasi-\mratms taken from~\cite{pzbpolc} that we shall need later. Moreover we present a proof for  Proposition~\ref{prop:utooptilat}.

\subsection{A property of quasi-\mratms}

We shall need the following property of quasi-\mratms which is proved in~\cite[Lemma~6.6]{pzbpolc}.

\begin{lemma}\label{lem:qmult}
  Let \Gs be a \pvari~and~\mbox{$\rho: 2^{A^*} \!\!\to R$} be a \Gs-\mratm. Given $H,L \subseteq A^*$, $q \in \opti{\Gs}{H,\rho}$ and $r \in \opti{\Gs}{L,\rho}$, we have $\quasir(qr) \in  	\opti{\Gs}{HL,\rho}$.
\end{lemma}

\subsection{Proof of Proposition~\ref{prop:utooptilat}}

Let us first recall the statement of Proposition~\ref{prop:utooptilat}. It is as follows.

\adjustc{prop:utooptilat}
\begin{proposition}
  Let \Cs be a \vari of group languages and \Ds a half level within the concatenation hierarchy of basis \Cs. Assume that there exists an algorithm which computes \ioptic{\lratauxd} from a morphism $\alpha$ and a \nice \mratm $\rho$. Then, \Ds-covering is decidable.
\end{proposition}
\restorec

We fix \Cs and \Ds as in the statement of the proposition. We show that given a morphism $\alpha$ and a \nice \mratm $\rho$, one may compute \popti{\Ds}{\alpha}{\rho} from its subset \ioptic{\lratauxd}. The result is then an immediate corollary of Proposition~\ref{prop:lreduc}.

We fix a morphism $\alpha: A^* \to M$ and a \nice \mratm $\rho: 2^{A^*} \to R$ for the proof. We let $S \subseteq M \times  R$ as the least subset of $R$ which satisfies the following properties:
\begin{enumerate}
\item We have $\ioptic{\lratauxd} \subseteq S$.
\item We have $(\alpha(w),\rho(w)) \in S$ for every $w \in A^*$.
\item We have $\dclosr S = S$.
\item For every $(s,r),(s',r') \in S$, we have $(ss',rr') \in S$.
\end{enumerate}
Clearly, one may compute $S$ from $\alpha$, $\rho$ and \ioptic{\lratauxd}. Hence, Proposition~\ref{prop:utooptilat} is now immediate from the following lemma.

\begin{lemma} \label{lem:utooptilat}
  We have $S = \popti{\Ds}{\alpha}{\rho}$.
\end{lemma}

We now concentrate on proving Lemma~\ref{lem:utooptilat}. For the proof, we let $L$ as a \Cs-optimal \iden for \lratauxd: $L \in \Cs$, $\veps \in L$ and $\ioptic{\lratauxd} = \lratauxd(L)$.

\medskip

We first prove that $S \subseteq \popti{\Ds}{\alpha}{\rho}$. This amounts to proving that \popti{\Ds}{\alpha}{\rho} satisfies the four properties in the definition of $S$. That $\ioptic{\lratauxd} \subseteq \popti{\Ds}{\alpha}{\rho}$ is immediate since $\ioptic{\lratauxd} = \lratauxd(L)$, $\popti{\Ds}{\alpha}{\rho} = \lratauxd(A^*)$, $L \subseteq A^*$ and \lratauxd is a \ratm. The other three properties are actually generic: they are satisfied whenever $\rho$ is a \mratm and \Ds is a \pvari (which is the case here since \Ds is a half level within the hierarchy of basis \Cs). We refer the reader to~\cite[Lemma~9.11]{pzcovering2} for the proof.

\medskip

We turn to the converse inclusion: $\popti{\Ds}{\alpha}{\rho} \subseteq S$. Consider $(s,r) \in \popti{\Ds}{\alpha}{\rho}$, we show that $(s,r) \in S$. The argument is based on the following lemma which is where we use the hypothesis that \Cs is a class of group languages (which implies that $L \in \Cs$ is recognized by a finite group).

\begin{lemma} \label{lem:grouplem2}
  There exist $\ell \in \nat$ and $\ell$ letters $a_1,\dots,a_\ell \in A$ such that $(s,r) \in \lratauxd(La_1L \cdots a_\ell L)$.
\end{lemma}

\begin{proof}
  Identical to the proof of Lemma~\ref{lem:grouplem} using the hypothesis that $L$ is a group language and \lratauxd is a \ratm.
\end{proof}

By definition \Ds is a half level within the hierarchy of basis \Cs. This means that \Ds is the polynomial closure of a \vari and by Theorem~\ref{thm:polclos}, we obtain that \Ds is a \pvari closed under concatenation and marked concatenation. Consequently, we obtain from Lemma~\ref{lem:copelrat} that \lratauxd is quasi-\tame and the associated endomorphism of $(2^{M \times R},\cup)$ is defined by $\quasi{\lratauxd}(T) = \dclosr T$ for all $T \in 2^{M \times R}$. Hence, by Axioms~\ref{itm:ax1} and~\ref{itm:ax3} in the definition of quasi-\mratms, the hypothesis that $(s,r) \in \lratauxd(La_1L \cdots a_\ell L)$ given by Lemma~\ref{lem:grouplem2} implies that,
\[
  (s,r) \in \dclosr\left(\lratauxd(L) \cdot \prod_{1 \leq i \leq \ell}(\lratauxd(a_i) \cdot \lratauxd(L))\right)
\]
By definition of $L$, we have $\lratauxd(L) = \ioptic{\lratauxd}$. Moreover, we know that $\ioptic{\lratauxd} \subseteq S$ by definition of $S$. Hence, since $S$ is also closed under multiplication and downset, that that $(s,r) \in S$ is now immediate from the following lemma.

\begin{lemma} \label{lem:pollet}
  For every $a \in A$, we have $\lratauxd(a) \subseteq S$.
\end{lemma}

\begin{proof}
  Consider an element $(t,q) \in \lratauxd(a)$. We show that $(t,q) \in S$. By definition of \lratauxd, we have $q \in \opti{\Ds}{\alpha\inv(t) \cap \{a\},\rho}$. This implies that $\alpha\inv(t)\cap\{a\} \neq \emptyset$ which means that $t = \alpha(a)$. Hence, we get,
  \[
    q \in \opti{\Ds}{\{a\},\rho}
  \]
  Recall that \Ds is a half level in the hierarchy of basis \Cs. We now consider two cases depending on whether this level is 1/2 (i.e. $\Ds = \pol{\Cs}$) or a strictly higher one. We start with the latter case which is simpler.

  \smallskip

  If \Ds is at least the level 3/2 in the hierarchy of basis \Cs. It is immediate from Fact~\ref{fct:finitelangs} that $\{a\} \in \Ds$.   Hence, $\{\{a\}\}$ is a \Ds-cover of $\{a\}$. Therefore, $\opti{\Ds}{\{a\},\rho} = \prin{\rho}{\{a\}} = \dclosr \{\rho(a)\}$. It follows that $q \leq \rho(a)$. By hypothesis, we know that $(\alpha(a),\rho(a))\in S$ and $\dclosr S = S$. Hence, since $t=\alpha(a)$, we conclude that $(t,q) \in S$ as desired.

  \smallskip

  We now assume that is the level 1/2: $\Ds = \pol{\Cs}$. Let $H$ be a \Cs-optimal \iden for $\rho$: $H \in \Cs$, $\veps \in H$ and $\ioptic{\rho} = \rho(H)$. Since \Ds is a half level within the hierarchy of basis \Cs, it contains \Cs and is closed under marked concatenation. Thus, we know that $HaH \in \Ds$. Moreover, we have $a \in HaH$ since $\veps \in H$. Hence, $\{HaH\}$ is a \Ds-cover of $\{a\}$ which means that $\opti{\Ds}{\{a\},\rho} \subseteq \prin{\rho}{\{HaH\}}$. Since $q \in \opti{\Ds}{\{a\},\rho}$, this yields $q \in \prin{\rho}{\{HaH\}}$ which means that,
  \[
    q \leq \rho(HaH) = \rho(H) \cdot \rho(a) \cdot \rho(H) = \ioptic{\rho} \cdot \rho(a) \cdot \ioptic{\rho}
  \]
  Since $t = \alpha(a)$ and $S$ is closed under downset by definition, it remains to show that $(\alpha(a), \ioptic{\rho} \cdot\rho(a)\cdot \ioptic{\rho})\in S$. This will imply as desired that $(t,q) \in S$.

  Since \Ds = \pol{\Cs}, Lemma~\ref{lem:polctoc} yields that $\ioptic{\rho} \in \opti{\Ds}{\{\veps\},\rho}$.   Moreover $\veps \in \alpha\inv(1_M) \cap L$ by definition of $L$. Hence, Fact~\ref{fct:linclus} yields that,
  \[
    \ioptic{\rho} \in \opti{\Ds}{\alpha\inv(1_M) \cap L,\rho}
  \]
  This exactly says that $(1_M,\ioptic{\rho}) \in \lratauxd(L)$ which is equal to $\ioptic{\lratauxd}$ by definition of $L$. Moreover, we know that $\ioptic{\lratauxd} \subseteq S$ by definition of $S$. Hence, we obtain that, $(1_M,\ioptic{\rho}) \in S$. Finally, we also know that $(\alpha(a),\rho(a)) \in S$ and that $S$ is closed under multiplication. Altogether, this yields,
  \[
    (\alpha(a), \ioptic{\rho} \cdot\rho(a)\cdot \ioptic{\rho})\in S
  \]
  This concludes the proof of Lemma~\ref{lem:pollet}.
\end{proof}

\section{Omitted proofs in Section~\ref{sec:bpolg}}
\label{app:bpolg}
This appendix is devoted to the proof of Theorem~\ref{thm:bpolg}. Let us first recall the statement. Recall that a \vari of group languages \Cs is fixed.

\adjustc{thm:bpolg}
\begin{theorem} 	Let $\rho: 2^{A^*} \to R$ be a \nice \mratm. Then, \bpoluopti is the greatest \bpol{\Cs}-complete subset of~$R$.
\end{theorem}
\restorec

We fix a \nice \mratm $\rho: 2^{A^*} \to R$ for the proof. As expected, the argument involves two independent directions.
\begin{itemize}
\item First, we show that for every \bpol{\Cs}-complete subset $S \subseteq R$, the inclusion $S \subseteq \bpoluopti$ holds. This corresponds to soundness in the greatest fixpoint procedure computing \bpoluopti (it only computes elements in \bpoluopti).
\item Then, we show that \bpoluopti itself is \bpol{\Cs}-complete. This corresponds to completeness in the greatest fixpoint procedure computing \bpoluopti (it computes all elements in \bpoluopti).
\end{itemize}

Both directions reuse material which was originally introduced in~\cite{pzbpolc}. We start with soundness.

\subsection{Soundness}

We have to show that for every \bpol{\Cs}-complete subset $S \subseteq R$, the inclusion $S \subseteq \bpoluopti$ holds. We first introduce an infinite sequence of auxiliary \ratms (defined from $\rho$). We then use them to formulate two independent statements. When put together they imply the desired result.

\smallskip

Using induction, we define a \ratm $\tau_n: 2^{A^*} \to Q_n$ for every $n \in \nat$. When $n = 0$, the rating set $Q_0$ is $(2^R,\cup)$ and $\tau_0$ is defined as follows,
\[
  \begin{array}{llll}
    \tau_0: & (2^{A^*},\cup) & \to     & (2^R,\cup)                    \\
            & K       & \mapsto & \{\rho(w) \mid w \in K\}
  \end{array}
\]
It is immediate by definition that $\tau_0$ is indeed a \ratm (i.e. a monoid morphism).

Assume now that $n \geq 1$ and that $\tau_{n-1}: 2^{A^*} \to Q_{n-1}$ is defined. Recall that $\rho_*:  A^* \to R$ denotes the canonical monoid morphism associated to the \mratm $\rho$. We define $\tau_n$ as the \ratm \lrataux{\pol{\Cs}}{\rho_*}{\tau_{n-1}} (see Section~\ref{sec:units} for the definition). In particular, this means that for all $n \geq 1$, the rating set $Q_n$ of $\tau_n$ is $(2^{R \times Q_{n-1}},\cup)$.

We complete this definition with maps $f_n: Q_n \to 2^R$ for $n \in \nat$. In this case as well, we use induction on $n$.
\begin{itemize}
\item For $n = 0$, let $T \in Q_0= 2^R$. Given $s \in R$, we have $s \in f_0(T)$ if there exists $r_1,\dots,r_k \in T$ such that $s \leq r_1 + \cdots + r_k$.
\item For $n \geq 1$, let $T \in Q_n = 2^{R \times Q_{n-1}}$. Given $s \in R$, we have $s \in f_n(T)$ if there exists  $(r_1,T_1),\dots,(r_k,T_k) \in T$ such that  $s \leq r_1 + \cdots + r_k$ and $r_1 + \cdots + r_k \in f_{n-1}(T_i)$ for every $i \leq k$.
\end{itemize}
We have the following simple fact which is immediate from the definition.

\begin{fact} \label{fct:bpol:fmapinc}
  For every $n \in \nat$ and $T,T' \in Q_n$ such that $U \subseteq U'$, we have $f_n(T) \subseteq f_n(T')$.
\end{fact}

This concludes our preliminary definitions. We now come back to the soundness proof. Using the objects that we just defined, we present two propositions. The first one is taken from~\cite{pzbpolc} while the second is new and specific to our setting.

\begin{proposition} \label{prop:sound1}
  Consider a language $L \in \pol{\Cs}$. Then, the following inclusion holds,
  \[
    \bigcap_{n \in \nat} f_n(\tau_n(L)) \subseteq \opti{\bpol{\Cs}}{L,\rho}
  \]
\end{proposition}

\begin{proposition} \label{prop:sound2}
  Let $S \subseteq R$ which is \bpol{\Cs}-complete for $\rho$. Then,
  \[
    S \subseteq \bigcap_{n \in \nat} f_n(\ioptic{\tau_n})
  \]
\end{proposition}

When combined, the two propositions imply the desired result: for every \bpol{\Cs}-complete subset $S \subseteq R$, the inclusion $S \subseteq \bpoluopti$ holds.  Indeed, consider a \bpol{\Cs}-complete subset $S \subseteq R$. We let $L \subseteq A^*$ as a \Cs-optimal \iden for \bratauxbc. That is, we have $L \in \Cs$, $\veps \in L$ and $\bratauxbc(L) = \bpoluopti$. Consequently, it now suffices to show that $S \subseteq \bratauxbc(L)$. Since $L \in \Cs$ and $\veps \in L$, it is immediate that $\ioptic{\tau_n} \subseteq \tau_n(L)$ for every $n \in \nat$. Thus, by Fact~\ref{fct:bpol:fmapinc}, $f_n(\ioptic{\tau_n}) \subseteq f_n(\tau_n(L))$ for every $n \in \nat$. Consequently, Proposition~\ref{prop:sound1} yields that,
\[
  \bigcap_{n \in \nat} f_n(\ioptic{\tau_n}) \subseteq \opti{\bpol{\Cs}}{L,\rho} = \bratauxbc(L)
\]
Finally, since $S$ is \bpol{\Cs}-complete by hypothesis, we may combine this inclusion with the one given by Proposition~\ref{prop:sound2} to obtain that $S \subseteq \bpoluopti$ as desired. This concludes the soundness proof. It now remains to prove the two propositions.

Proposition~\ref{prop:sound1} is proved in~\cite[Proposition~8.2]{pzbpolc}. This is actually a generic result which holds for any class built using Boolean closure (Proposition~\ref{prop:sound1} remains true when replacing \pol{\Cs} by an arbitrary lattice \Gs in the definition of the \ratms $\tau_n$ and in the statement). Here we focus on proving Proposition~\ref{prop:sound2} which is specific to our setting: we require the hypothesis that we are considering the class \bpol{\Cs} for \Cs a \vari of group languages. The remainder of this subsection is devoted to this proof.

\medskip
\noindent
{\bf Preliminaries.} We start by presenting a few additional properties of the \ratms $\tau_n$ (again they are taken from~\cite{pzbpolc}). We know that these \ratms are quasi-\tame.

First observe that since $R$ is a semiring, it is immediate that we may lift its multiplication inductively to all rating sets $Q_n$ in the natural way. This yields semirings $(Q_n,\cup,\cdot)$ for all $n \in \nat$. We shall rely heavily on these multiplications. However, it is important to keep in mind that except for $\tau_0$, the \ratms $\tau_n$ are \textbf{not} \tame. Instead, there are only quasi-\tame. We state this property in the following lemma. The proof is a simple induction on $n \in \nat$ and is based on a generalization of Lemma~\ref{lem:copelrat} (see~\cite[Lemma~8.8]{pzbpolc})

\begin{lemma} \label{lem:squasitame}
  The \ratm $\tau_0: 2^{A^*} \to Q_0$ is \tame. Moreover, for every $n \geq 1$, the \ratm $\tau_n: 2^{A^*} \to Q_n$ is quasi-\tame and the associated endomorphism $\quasi{\tau_n}$ of $(Q_n,\cup)$ is as follows. For every $U \in Q_n = 2^{T \times Q_{n-1}}$,
  \[
    \quasi{\tau_n}(T) = \dclosp{Q_{n-1}} \{(r,\quasi{\tau_{n-1}}(V)) \mid (r,V) \in T\}
  \]
\end{lemma}

We complete this result with two lemmas which connect the hypothesis that the \ratms $\tau_n$ are quasi-\tame with the maps $f_n: Q_n \to 2^R$ (the proofs are straightforward and available in~\cite[Lemmas~8.9 and~8.10]{pzbpolc}).

\begin{lemma} \label{lem:fmapmu}
  For every $n \in \nat$ and $T \in Q_n$, we have $f_n(T) \subseteq f_n(\quasi{\tau_n}(T))$.
\end{lemma}

\begin{lemma} \label{lem:bpol:fmapmult}
  For every $n \in \nat$ and $U,U' \in Q_n$, we have $f_n(U) \cdot f_n(U') \subseteq  f_n(U \cdot U')$.
\end{lemma}

\medskip
\noindent
{\bf Proof of Proposition~\ref{prop:sound2}.} We fix a \bpol{\Cs}-complete subset $S\subseteq R$ for the proof. We have to show that for every $n \in \nat$, we have $S \subseteq f_n(\ioptic{\tau_n})$. We fix $n \in \nat$ and $s \in S$: we show that $s \in f_n(\ioptic{\tau_n})$. The proof is an induction on $n \in \nat$. We start with some definitions which are common to the base case ($n = 0$) and the inductive step ($n \geq 1$).

Since $S$ is \bpol{\Cs}-complete by hypothesis and $s \in S$, we have $(r_1,U_1),\dots,(r_k,U_k) \in \ioptic{\bwmrats}$ which satisfy~\eqref{eq:bpolcarac}: $s \leq r_1 + \cdots + r_k$ and $r_1 + \cdots + r_k \in \dclosr U_i$ for every $i \leq k$. We now consider two cases depending on whether $n = 0$ or $n \geq 1$.

\medskip
\noindent
{\bf Base case: $n=0$.} Let $L$ be a \Cs-optimal \iden for $\tau_0$: we have $L \in \Cs$, $\veps \in L$ and $\ioptic{\tau_0} = \tau_0(L)$. We now have to prove that $s \in f_0(\tau_0(L))$. We have $(r_i,U_i) \in \ioptic{\bwmrats}$ for every $i \leq k$. Since $L \in \Cs$ satisfies $\veps \in L$, we know that $\ioptic{\bwmrats} \subseteq \bwmrats(L)$ by definition. Consequently, we have $(r_i,U_i) \in \bwmrats(L)$ for every $i \leq k$.

By definition of \bwmrats, it follows that we have $w_1,\dots,w_k \in L$ such that $r_i = \rho(w_i)$ for every $i \leq k$. It follows that $r_1,\dots,r_k \in \tau_0(L)$ by definition of $\tau_0$. Hence, since $s \leq r_1 + \cdots + r_k$, we obtain that $s \in f_0(\tau_0(L))$ by definition of $f_0$ which concludes the proof in this case.

\medskip
\noindent
{\bf Inductive step: $n \geq 1$.} The argument is based on the following lemma which is proved using induction on $n$.

\begin{lemma} \label{lem:soundopti}
  For every $(r,U) \in \ioptic{\bwmrats}$, there exists $T \in Q_{n-1}$ such that $(r,T) \in \ioptic{\tau_n}$ and $\dclosr U \subseteq f_{n-1}(T)$.
\end{lemma}

We first use Lemma~\ref{lem:soundopti} to show that $s \in f_n(\ioptic{\tau_n})$ and complete the main proof. For every $i \leq k$, we know that $(r_i,U_i) \in \ioptic{\bwmrats}$ and $r_1 + \cdots + r_k \in \dclosr U_i$. Therefore, Lemma~\ref{lem:soundopti} yields that for every $i \leq k$, we have $T_i \in Q_{n-1}$ such that $(r_i,T_i) \in \ioptic{\tau_n}$ and $r_1 + \cdots + r_k \in f_{n-1}(T_i)$. Since we have $s \leq r_1 + \cdots + r_k$ this yields that $s \in f_{n}(\ioptic{\tau_n})$ by definition of $f_n$. This concludes the main proof. It remains to prove Lemma~\ref{lem:soundopti}.

\begin{proof}[Proof of Lemma~\ref{lem:soundopti}]
  Consider $(r,U) \in \ioptic{\bwmrats}$. We have to exhibit the appropriate $T \in Q_{n-1}$ such that,
  \begin{equation} \label{eq:bsound}
    (r,T) \in \ioptic{\tau_n} \quad \text{and} \quad \dclosr U \subseteq f_{n-1}(T)
  \end{equation}
  For the proof, we let $L$ as a \Cs-optimal \iden for $\tau_n$: we have $L \in \Cs$, $\veps \in L$ and $\ioptic{\tau_n} = \tau_n(L)$.

  Since $L \in \Cs$ and $\veps \in L$, we know that $\ioptic{\bwmrats} \subseteq \bwmrats(L)$. Therefore, $(r,U) \in \bwmrats(L)$. By definition of \bwmrats, this yields $w \in L$ such that $\bwmrats(w) = \{(r,U)\}$. We now consider two cases depending on whether $w = \veps$ or $w \in A^+$.

  \medskip
  \noindent
  {\bf Case~1:} Assume first that $w = \veps$. It follows that $(r,U) = (1_R,\{1_R\})$ (since $\{(1_R,\{1_R\})\}$ is the neutral element of $2^{R \times 2^R}$). We prove that~\eqref{eq:bsound} holds for $T = \tau_{n-1}(\veps) \in Q_{n-1}$.

  We first show that $(r,T) \in \ioptic{\tau_n}$. Since $\veps \in L$ by definition, we have $\tau_n(\veps) \subseteq \tau_n(L) = \ioptic{\tau_n}$. Moreover,
  \[
    \begin{array}{lll}
      \tau_n(\veps) & = & \lrataux{\pol{\Cs}}{\rho_*}{\tau_{n-1}}(\veps)                        \\
                    & = & \{(q,T') \mid T' \in \opti{\pol{\Cs}}{\{\veps\} \cap \rho_*\inv(q),\tau_{n-1}}\}
    \end{array}
  \]
  One may verify that this yields $(r,T) = (1_R,\tau_{n-1}(\veps)) \in \tau_n(\veps)$. Altogether, we get $(r,T) \in \ioptic{\tau_n}$ as desired.

  It remains to show that $\dclosr U \subseteq f_{n-1}(T)$. One may verify that $1_R = \rho(\veps) \in f_{n-1}(\tau_{n-1}(\veps))$ (this follows from a simple induction on $n$). Thus, since $U = \{1_R\}$ and $T = \tau_{n-1}(\veps)$, we get $U \subseteq f_{n-1}(T)$. This implies $\dclosr U \subseteq f_{n-1}(T)$ since $\dclosr f_{n-1}(T) = f_{n-1}(T)$ by definition of $f_{n-1}$.

  \medskip
  \noindent
  {\bf Case~2:} Assume now that $w \in A^+$. In that case, we have $\ell \geq 1$ and $\ell$ letters $a_1,\dots,a_\ell$ such that $w = a_1 \cdots a_\ell$. Let $V = \ioptic{\tau_{n-1}} \in Q_{n-1}$. We show that~\eqref{eq:bsound} holds for,
  \[
    T = \quasi{\tau_{n-1}}(V \cdot \tau_{n-1}(a_1) \cdot V^2 \cdot \tau_{n-1}(a_2) \cdots V^2 \cdot \tau_{n-1}(a_\ell) \cdot V)
  \]
  We first show that $(r,T) \in \ioptic{\tau_n}$.	Clearly, we have $\tau_{n-1}(a) \in \opti{\pol{\Cs}}{\{a\},\tau_{n-1}}$ for every $a \in A$. Moreover, $V = \ioptic{\tau_{n-1}} \in \opti{\pol{\Cs}}{\{\veps\},\tau_{n-1}}$ by Lemma~\ref{lem:polctoc}. Therefore, since $w = a_1 \cdots a_\ell$,	it is immediate by definition of $T$ and Lemma~\ref{lem:qmult} that,
  \[
    T \in \opti{\pol{\Cs}}{\{w\},\tau_{n-1}}
  \]
  Moreover, $r = \rho(w)$ and $w \in L$. Thus, $w \in \rho_*\inv(r) \cap L$. By Fact~\ref{fct:linclus}, this yields,
  \[
    T \in \opti{\pol{\Cs}}{\rho_*\inv(r) \cap L,\tau_{n-1}}
  \]
  By definition, $\tau_n$ is the \ratm \lrataux{\pol{\Cs}}{\rho_*}{\tau_{n-1}}. Hence, the above yields $(r,T) \in \tau_n(L)$. By choice of $L$, this exactly says that $(r,T) \in \ioptic{\tau_n}$.

  It remains to show that $\dclosr U \subseteq f_{n-1}(T)$. Since $\bwmrats(w) = \{(r,U)\}$, it follows from the definition of \bwmrats that $r = \rho(w)$ and $U$ is equal to the following set,
  \begin{equation} \label{eq:bwminc}
    U = S \cdot \{\rho(a_1)\} \cdot S^2 \cdot \{\rho(a_2)\} \cdots S^2 \cdot \{\rho(a_\ell)\} \cdot S
  \end{equation}
  Since $V = \ioptic{\tau_{n-1}}$, we know by induction hypothesis in Proposition~\ref{prop:sound2} that $S \subseteq f_{n-1}(V)$. Moreover, it can be verified using a simple induction on $n$ that for every $a \in A$, we have  $\rho(a) \in f_{n-1}(\tau_{n-1}(a))$. Hence, we may combine~\eqref{eq:bwminc} with Lemma~\ref{lem:bpol:fmapmult} to obtain,
  \[
    U \subseteq f_{n-1}(V \cdot \tau_{n-1}(a_1) \cdot V^2 \cdot \tau_{n-1}(a_2) \cdots V^2 \cdot \tau_{n-1}(a_\ell) \cdot V)
  \]
  By definition of $T$, we then obtain from Lemma~\ref{lem:fmapmu} that,
  \[
    U \subseteq f_{n-1}(T)
  \]
  Since $\dclosr f_{n-1}(T) = f_{n-1}(T)$ by definition of $f_{n-1}$, this yields that $\dclosr U \subseteq f_{n-1}(T)$, finishing the proof.
\end{proof}

\subsection{Completeness}

We turn to the completeness direction in the proof of Theorem~\ref{thm:bpolg}. Recall that a \nice \mratm $\rho: 2^{A^*} \to R$ is fixed. We need to show that  \ioptic{\bratauxbc} is \bpol{\Cs}-complete.

In the whole section, we shall have to deal with the \ratm \bratauxbc. For the sake of avoiding clutter we denote it by $\tau$. Recall that the definition is as follows,
\[
  \begin{array}{llll}
    \tau: & 2^{A^*} & \to     & 2^R                       \\
          & K       & \mapsto & \opti{\bpol{\Cs}}{K,\rho}
  \end{array}
\]
Our objective is to show that $\ioptic{\tau} \subseteq R$ is \bpol{\Cs}-complete. For the sake of avoiding clutter, we write $S = \ioptic{\tau}$.

The completeness proof is now based on two propositions which are proved independently. The first one is taken from~\cite{pzbpolc} while the second is new.

\begin{proposition} \label{prop:bpol:comp1}
  For every $L \in \bpol{\Cs}$ and $s \in \tau(L)$, we have $r_1,\dots,r_k \in R$ such that $(r_i,\{r_1+ \cdots +r_k\}) \in \lrataux{\pol{\Cs}}{\rho_*}{\tau}(L)$ for every $i \leq k$ and $s \leq r_1 + \cdots + r_k$.
\end{proposition}

For the second proposition, we let $L'$ as a \Cs-optimal \iden for \bwmrats: we have $L' \in \Cs$, $\veps \in L'$ and $\ioptic{\bwmrats}=\bwmrats(L')$. Furthermore, we define $L =L' \cap A^+$. Note that $L \in \pol{\Cs} \subseteq \bpol{\Cs}$ since $L' \in \Cs \subseteq \pol{\Cs}$ and $A^+ = \bigcup_{a \in A} A^*aA^* \in \pol{\Cs}$. We shall later apply Proposition~\ref{prop:bpol:comp1} for this particular language $L$.

\begin{proposition} \label{prop:bpol:comp2}
  For every $(r,U) \in \lrataux{\pol{\Cs}}{\rho_*}{\tau}(L)$, there exists $V \in 2^R$ such that $(r,V) \in \ioptic{\bwmrats}$ and $U \subseteq\dclosr V$.
\end{proposition}

Let us first use the two propositions to finish the completeness proof. We show that $S = \ioptic{\tau}$ is \bpol{\Cs}-complete for $\rho$. Let $s \in S$. We have to exhibit $(r_1,U_1),\dots,(r_k,U_k) \in \ioptic{\bwmrats}$ such that~\eqref{eq:bpolcarac} is satisfied: $s \leq r_1 + \cdots + r_k$ and $r_1 + \cdots + r_k \in\ \dclosr U_i$ for every $i \leq k$.

By definition, $S = \ioptic{\tau}$. Thus, since $L' \in \Cs$ and $\veps \in L'$, this implies that $S \subseteq \tau(L')$ and we obtain that $s \in \tau(L')$. Moreover, $L' = L \cup \{\veps\}$ by definition. Therefore, we have $s \in \tau(L) \cup \tau(\veps)$. We distinguish two cases depending on which set of this union $s \in S$ belongs to.

Assume first that $s \in \tau(\veps)$, i.e. $s \in \opti{\bpol{\Cs}}{\{\veps\},\rho}$ by definition of $\tau$. Since $\{\veps\} \in \bpol{\Cs}$ by Fact~\ref{fct:finitelangs}, $\{\{\veps\}\}$ is a \bpol{\Cs}-cover of $\{\veps\}$. Thus, we have $\opti{\bpol{\Cs}}{\{\veps\},\rho} \subseteq \prin{\rho}{\{\{\veps\}\}} = \dclosr \{\rho(\veps)\}$. It follows that $s \leq \rho(\veps) = 1_R$ (recall that $\rho$ is a \mratm). Clearly, $(1_R,\{1_R\}) \in \ioptic{\bwmrats}$ (since $(1_R,\{1_R\})$ is the neutral element of $R \times 2^R$, we have $\bwmrats(\veps) = (1_R,\{1_R\})$). Hence it is immediate that~\eqref{eq:bpolcarac} is satisfied for $k = 1$ and $(r_1,U_1) = (1_R,\{1_R\}) \in \ioptic{\bwmrats}$: we have $s \leq 1_R$ and $1_R \in \dclosr U_1$.

Assume now that $s \in \tau(L)$.  In that case, since $L \in \bpol{\Cs}$, we obtain from Proposition~\ref{prop:bpol:comp1} that there exist $r_1,\dots,r_k \in R$ such that $(r_i,\{r_1+ \cdots +r_k\}) \in \lrataux{\pol{\Cs}}{\rho_*}{\tau}(L)$ for every $i \leq k$ and $s \leq r_1 + \cdots + r_k$. We may now apply Proposition~\ref{prop:bpol:comp2} which yields that for every $i \leq k$, there exists $U_i \in 2^R$ such that $(r_i,U_i) \in \ioptic{\bwmrats}$ and $r_1+ \cdots +r_k \in \dclosr U_i$. Altogether, we get that~\eqref{eq:bpolcarac} holds, finishing the proof.

\medskip

It remains to prove the propositions. Similarly to what happened for the soundness direction, Proposition~\ref{prop:bpol:comp1} is actually a generic result which holds for any class built using Boolean closure (it remains true when replacing \pol{\Cs} by an arbitrary lattice \Gs in the definition of the \ratm $\tau$ and in the statement). It is proved in~\cite[Proposition~9.1]{pzbpolc}. An important point is that Proposition~\ref{prop:bpol:comp1} is where we require the hypothesis that $\rho$ is \nice. Here, we show Proposition~\ref{prop:bpol:comp2} which is specific to our setting: we need the hypothesis that we are considering the class \bpol{\Cs} for \Cs a \vari of group languages. The remainder of this subsection is devoted to this proof.

\medskip
\noindent
{\bf Proof of Proposition~\ref{prop:bpol:comp2}.} Consider $(r,U) \in \lrataux{\pol{\Cs}}{\rho_*}{\tau}(L)$. Our goal is to exhibit $V \in 2^R$ such that,
\begin{equation} \label{eq:cgoal}
  (r,V) \in \ioptic{\bwmrats} \text{ and } U \subseteq \dclosr V
\end{equation}

By definition of the \ratm \lrataux{\pol{\Cs}}{\rho_*}{\tau}, our hypothesis means that we have $U \in \opti{\pol{\Cs}}{\rho_*\inv(r) \cap L,\tau}$. We build a well-chosen \pol{\Cs}-cover of $\rho_*\inv(r) \cap L$ in order to use this hypothesis. The construction is based on Lemma~\ref{lem:half:pump}.

Let $H$ be a \Cs-optimal \iden for $\tau = \bratauxbc$: we have $H \in \Cs$, $\veps \in H$ and $S = \ioptic{\tau} = \tau(H)$. For every $w \in A^*$, we write $F_{w}$ for the language $H a_1 H  \cdots a_n H$ where $w = a_1 \cdots a_n$ (in particular, $F_{\veps} = H$). Note that by definition, $F_{w} \in \pol{\Cs}$ for every $w \in A^*$. We have the following immediate corollary of Lemma~\ref{lem:half:pump}.

\begin{lemma} \label{lem:coverpol2}
  There exists a \textbf{finite} subset $G \subseteq \rho_*\inv(r) \cap L$ such that $\{F_{w} \mid w \in G\}$ is a cover of $\rho_*\inv(r) \cap L$
\end{lemma}

In view of Lemma~\ref{lem:coverpol2}, we let $\Kb = \{F_{w} \mid w \in G\}$ which is a \pol{\Cs}-cover of $\rho_*\inv(r) \cap L$ by definition. Therefore, we have,
\[
  \opti{\pol{\Cs}}{\rho_*\inv(r) \cap L,\tau} \subseteq \prin{\tau}{\Kb}
\]
Consequently, we know that $U \in \prin{\tau}{\Kb}$ which yields $w \in G$ such that $U \subseteq \tau(F_w)$.

By construction, we have $G \subseteq \rho_*\inv(r) \cap L$ and $L = L' \cap A^+$. Thus, $w \in A^+$: we have $n \geq 1$ and $n$ letters $a_1, \dots, a_n$ such that $w = a_1 \cdots a_n$. We let,
\[
  V = S \cdot \{\rho(a_1)\} \cdot S^2 \cdot \{\rho(a_2)\} \cdots S^2 \cdot \{\rho(a_n)\} \cdot S
\]
It remains to show that~\eqref{eq:cgoal} holds: we have $(r,V) \in \ioptic{\bwmrats}$ and $U \subseteq \dclosr V$. Let us start with the former. By definition, $\rho(w) = r$ since $w \in G \subseteq \rho_*\inv(r) \cap L$. Thus, since $w = a_1 \cdots a_n$ it is immediate from the definition of \bwmrats that $\bwmrats(w) = \{(r,V)\}$. Moreover, we have $w \in G \subseteq \rho_*\inv(r) \cap L \subseteq L'$ and $\bwmrats(L') = \ioptic{\bwmrats}$ by definition of $L'$. It follow that $\bwmrats(w) \subseteq \ioptic{\bwmrats}$. Altogether, we get $(r,V) \in \ioptic{\bwmrats}$.

Finally, we prove that $U \subseteq \dclosr V$. By definition, $F_w = H a_1 H  \cdots a_n H$.  Moreover, since $\veps \in H$, we have $H \subseteq H^2$. Thus,
\[
  F_{w} = H a_1 H  \cdots a_n H \subseteq H a_1 H^2  \cdots H^2a_n H
\]
Therefore, since we have $U \subseteq \tau(F_w)$, we also get $U \subseteq \tau(H a_1 H^2  \cdots H^2a_n H)$. Moreover, Lemma~\ref{lem:copelbrat} yields that the \ratm $\tau = \bratauxbc$ is \pol{\Cs}-\tame and the associated endomorphism of $(2^R,\cup)$ is defined by $\quasi{\tau}(W) = \dclosr W$. Hence, since $H \in \Cs \subseteq \pol{\Cs}$ and $\tau(H) = S$, it follows from Axiom~\ref{itm:ax3} in the definition of \pol{\Cs}-\tame \ratms that,
\[
  U  \subseteq \dclosr (S \cdot \tau(a_1) \cdot S^2 \cdots S^2 \cdot \tau(a_n) \cdot S)
\]
Additionally, observe that we have the following fact since $\tau = \bratauxbc$ by definition.

\begin{fact} \label{fct:cletters}
  For every $a \in A$, we have $\tau(a) \subseteq \dclosr \{\rho(a)\}$.
\end{fact}

\begin{proof}
  By definition of $\tau$, $\tau(a) = \opti{\bpol{\Cs}}{\{a\},\rho}$. Moreover, we have $\{a\} \in \bpol{\Cs}$ by Fact~\ref{fct:finitelangs}. Therefore, it is immediate that $\Kb = \{\{a\}\}$ is \bpol{\Cs}-cover of $\{a\}$. Consequently, $\tau(a) \subseteq \prin{\rho}{\Kb} =  \dclosr \{\rho(a)\}$.
\end{proof}

In view of Fact~\ref{fct:cletters}, we now obtain that,
\[
  U  \subseteq \dclosr (S \cdot \{\rho(a_1)\} \cdot S^2 \cdots S^2 \cdot \{\rho(a_n)\} \cdot S) = \dclosr V
\]
This concludes the proof.

\section{Omitted proofs in Section~\ref{sec:pbpolg}}
\label{app:pbpolg}
This appendix is devoted to the proof of Theorem~\ref{thm:pbpolg}. The statement is as follows. Recall that a \vari of group languages \Cs is fixed.

\adjustc{thm:pbpolg}
\begin{theorem}
  Fix a morphism $\alpha: A^* \to M$ and a \nice \mratm $\rho: 2^{A^*} \to R$. Then, \pbpoluopti is the least \pbpol{\Cs}-complete subset of $M \times R$.
\end{theorem}
\restorec

For the proof, we fix a morphism $\alpha: A^* \to M$ and a \nice \mratm $\rho: 2^{A^*} \to R$. We need to show that \ioptic{\lratauxppc} is the least \pbpol{\Cs}-complete subset of $M \times R$. The proof involves two directions which are proved independently:
\begin{itemize}
\item First, we prove that the set \pbpoluopti is \pbpol{\Cs}-complete. This corresponds to soundness of the least fixpoint procedure computing \ioptic{\lratauxppc}. The argument is based on results of~\cite{pzgenconcat} about the polynomial closure operation.
\item Then, we prove that \pbpoluopti is included in every \pbpol{\Cs}-complete subset. This corresponds to completeness of the least fixpoint procedure computing \ioptic{\lratauxppc}. The argument is based on on the main theorem of~\cite{pseps3j}.
\end{itemize}

\subsection{Soundness}

We prove that \ioptic{\lratauxppc} is \pbpol{\Cs}-complete. For the sake of avoiding clutter, we shall write $\tau: 2^{A^*} \to 2^{M \times R}$ for the \ratm \lratauxppc. Moreover, we write $S = \ioptic{\tau}$: we have to show that $S$ is \pbpol{\Cs}-complete. For the proof we let $L$ as a \Cs-optimal \iden for $\tau$: $L \in \Cs$, $\veps \in L$ and $S = \tau(L)$. Moreover, since \Cs is a \vari, it follows from the next lemma that we may choose the language $L$ so that $LL=L$. This will be useful.

\begin{lemma}\label{lem:epswitbis}
  There exists a \Cs-optimal \iden $L$ for $\tau$ such that $LL=L$.
\end{lemma}

\begin{proof}
  Using Lemma~\ref{lem:epswit}, we get a \Cs-optimal \iden $H$ for $\tau$. Since \Cs is a \vari, it is standard that there exists a morphism $\beta: A^* \to N$ recognizing $H$ and such that every language recognized by $\beta$ belongs to \Cs (it suffices to choose $\beta$ as the \emph{syntactic morphism} of $H$, see~\cite{pingoodref} for example). We let $L = \beta\inv(1_N)$. By hypothesis, we have $\veps \in L$ and $L \in \Cs$. Moreover, $L \subseteq H$ since $\veps \in H$ and $H$ is recognized by $\beta$. Thus, $\tau(L) \subseteq \tau(H)$ and since $H$ was a \Cs-optimal \iden for $\tau$, $L$ is one as well. It remains to show that $LL = L$. The inclusion $L \subseteq LL$ is immediate since $\veps \in L$. Additionally, if $w \in LL$, then $w = uv$ with $\beta(u) = \beta(v) = 1_N$ which implies that $\beta(w) = 1_N$, i.e. $w \in L$. Thus, $LL \subseteq L$ which concludes the proof.
\end{proof}

Recall that since \pbpol{\Cs} is a \pvari closed under concatenation and marked concatenation, the following fact which is immediate from Lemma~\ref{lem:copelrat}.

\begin{fact} \label{fct:pbpquasi}
  The \ratm $\tau$ is quasi-\tame and the associated endomorphism $\quasi{\tau}$ of $(2^{M \times R},\cup)$ is defined by $\quasi{\tau}(T) = \dclosr T$ for every $T \in 2^{M \times R}$.
\end{fact}

Moreover, the following property can be verified using Fact~\ref{fct:finitelangs} (which implies that \pbpol{\Cs} contains every finite language).

\begin{fact} \label{fct:taufinite}
  For every $w \in A^*$, $\tau(w) = \dclosr \{(\alpha(w),\tau(w))\}$.
\end{fact}

It now remains to prove that $S= \tau(L)$ is \pbpol{\Cs}-complete. This involves four properties. We start with downset.

\medskip
\noindent
{\bf Downset.} We show that $\dclosr S = S$. The right to left inclusion is trivial. Hence, we focus on the converse one. Consider $(s,r) \in \dclosr S = \dclosr \tau(L)$. We have $r' \in R$ such that $r \leq r'$ and $(s,r') \in \tau(L)$. Since $\tau$ is an alias for \lratauxppc, this exactly says that,
\[
  r' \in \opti{\pbpol{\Cs}}{L \cap \alpha\inv(s),\rho}
\]
Since \imprints are closed under downset by definition, we get $r \in \opti{\pbpol{\Cs}}{L \cap \alpha\inv(s),\rho}$ and we conclude that $(s,r) \in \tau(L) = S$, finishing the proof.

\medskip
\noindent
{\bf Multiplication.} Consider $(s,r),(s',r') \in S = \tau(L)$, we have to show that $(ss',rr') \in S = \tau(L)$. We have,
\[
  (ss',rr') \in \tau(L) \cdot \tau(L)
\]
Since $\tau$ is quasi-\tame and the associated endomorphism of $(2^{M \times R},\cup)$ is $T \mapsto \dclosr T$, the axioms of quasi-\mratms yield that $\tau(LL) = \dclosr(\tau(L) \cdot \tau(L))$. We get that $(ss,rr') \in \tau(LL)$. Since $LL = L$ by definition of $L$, this yields $(ss',rr') \in \tau(L)$ as desired.

\medskip

It remains to handle the last two properties in the definition of \pbpol{\Cs}-complete subsets: \Cs-operation and \pbpol{\Cs}-operation.  First, we prove the following lemma which we shall need for both properties.

\begin{lemma} \label{lem:pbp:pbpolsound}
  For every $(r,T) \in \ioptic{\pbwmrats}$, we have,
  \[
    T \in \opti{\pol{\Cs}}{L \cap \rho_*\inv(r),\tau}
  \]
\end{lemma}

\begin{proof}
  Consider $(r,T) \in \ioptic{\pbwmrats}$. By definition of \pbwmrats, we have $w \in L$ such that $\pbwmrats(w) = \{(r,T)\}$. We consider two cases depending on whether $w = \veps$ or $w \in A^+$.

  \medskip
  \noindent
  {\bf Case~1:} Assume first that $w = \veps$. It follows that $(r,T) = (1_R,\{(1_M,1_R)\})$ (since $\{(1_R,\{(1_M,1_R)\})\}$ is the neutral element of $2^{R \times 2^{M \times R}}$). It follows that $L \cap \rho_*\inv(r) =  \{\veps\}$ (recall that $\veps \in L$ by definition). Thus, we have to show that,
  \[
    \{(1_M,1_R)\} \in \opti{\pol{\Cs}}{\{\veps\},\tau}
  \]
  By Fact~\ref{fct:taufinite}, we have $(1_M,1_R) \in \tau(\veps)$. Thus, since it is clear that $\tau(\veps) \in \opti{\pol{\Cs}}{\{\veps\},\tau}$, we get $\{(1_M,1_R)\} \in \opti{\pol{\Cs}}{\{\veps\},\tau}$ which concludes this case.

  \medskip
  \noindent
  {\bf Case~2:} Assume now that $w \in A^+$. In that case, we have $\ell \geq 1$ and $\ell$ letters $a_1,\dots,a_\ell$ such that $w = a_1 \cdots a_\ell$.

  Since $\pbwmrats(w) = \{(r,T)\}$, it follows from the definition of \pbwmrats that $r = \rho(w)$ and $T$ is equal to the following set,
  \[
    S \cdot \{(\alpha(a_1),\rho(a_1))\} \cdot S^2 \cdots S^2 \cdot \{(\alpha(a_\ell),\rho(a_\ell))\} \cdot S
  \]
  By Fact~\ref{fct:taufinite}, we have $\{(\alpha(a),\rho(a))\} \subseteq \tau(a)$ for every $a \in A$. Therefore, we obtain that,
  \[
    T \subseteq S \cdot \tau(a_1) \cdot S^2 \cdots S^2 \cdot \tau(a_\ell) \cdot S
  \]
  It is clear that, $\tau(a) \in \opti{\pol{\Cs}}{\{a\},\tau}$ for every $a \in A$. Moreover, $S = \ioptic{\tau} \in \opti{\pol{\Cs}}{\{\veps\},\tau}$ by Lemma~\ref{lem:polctoc}. Therefore, since $w = a_1 \cdots a_\ell$,	it is immediate from the above and Lemma~\ref{lem:qmult} that,
  \[
    T \in \opti{\pol{\Cs}}{\{w\},\tau}
  \]
  Moreover, since $r = \rho(w)$ and $w \in L$, we have $w \in \rho_*\inv(r) \cap L$. Thus, Fact~\ref{fct:linclus} yields that $T \in \opti{\pol{\Cs}}{\rho_*\inv(r) \cap L,\tau}$, concluding the proof.
\end{proof}

We may now handle the last two operations in the definition of \pbpol{\Cs}-complete subsets.

\medskip
\noindent
{\bf \Cs-operation.} Consider $(r,T) \in \ioptic{\pbwmrats}$. we have to show that $T \subseteq S = \tau(L)$. By Lemma~\ref{lem:pbp:pbpolsound}, we have,
\[
  T \in \opti{\pol{\Cs}}{L \cap \rho_*\inv(r),\tau}
\]
Since $L \in \Cs \subseteq \pol{\Cs}$, $\{L\}$ is a \pol{\Cs}-cover of $L \cap \rho_*\inv(r)$. Therefore, it follows that $T \subseteq \tau(L) = S$, concluding the proof.

\medskip
\noindent
{\bf \pbpol{\Cs}-operation.} Consider $(r,T) \in \ioptic{\pbwmrats}$ and a pair of idempotents $(e,f) \in \dclosr T \subseteq M \times R$. We have to show that,
\[
  (e,f + frf) = (e, f \cdot (1_R + r) \cdot f) \in \tau(L)
\]
Recall that $\tau = \lratauxppc$. Hence, this amounts to proving that $frf + f \in \opti{\pbpol{\Cs}}{L \cap \alpha\inv(e),\rho}$. Thus, we fix an arbitrary \pbpol{\Cs}-cover \Kb of $L \cap \alpha\inv(e)$ for the proof and show that $frf+f \in \prin{\rho}{\Kb}$.

First, we use \Kb to construct two auxiliary classes of languages \Ds and \Fs that we shall need for the proof (the argument reuses some standard results which are detailed in~\cite{pzgenconcat}).

\begin{fact} \label{fct:pbp:strat}
  There exist two finite \pvaris \Ds and \Fs which satisfy the following properties:
  \begin{itemize}
  \item $L \in \Ds$ and  $\Ds \subseteq \pol{\Cs}$, and,
  \item $\Fs \subseteq \pol{\cocl{\Ds}} \subseteq \pbpol{\Cs}$ and $K \in \Fs$ for every $K \in \Kb$.
  \end{itemize}
\end{fact}

\begin{proof}
  By definition \Kb contains finitely many languages in \pbpol{\Cs}. Moreover, it is shown in~\cite[Lemma~33]{pzgenconcat} that we have $\pbpol{\Cs} = \pol{\copol{\Cs}}$ where \copol{\Cs} denotes the class containing the languages with a complement in \pol{\Cs}. Consequently, there are finitely many languages in \pol{\Cs} such that every $K \in \Kb$ is built by applying unions and of marked concatenations to complements of these languages. It then follows from~\cite[Lemma~17]{pzgenconcat} that we may build a finite \pvari $\Ds \subseteq \pol{\Cs}$ which contains all these languages as well as $L \in \Cs \subseteq \pol{\Cs}$. Thus, by definition, we have $K \in \pol{\cocl{\Ds}}$ for every $K \in \Kb$. Finally, since \Kb is finite, we may reuse~\cite[Lemma~17]{pzgenconcat} to build another finite \pvari \Fs such that $\Fs \subseteq \pol{\cocl{\Ds}}$ and $K \in \Fs$ for every $K \in \Kb$. This concludes the proof.
\end{proof}

We fix \Ds and \Fs as the finite \pvaris described in Fact~\ref{fct:pbp:strat} for the remainder of the proof. Since \Ds and \Fs are finite, it is classical to associate relations \canod and \canof over $A^*$. Given $w,w' \in A^*$, we have, write $w \canod w'$ if and only if the following holds:
\[
  \begin{array}{lll}
    w \canod w' & \text{if and only if} & \text{$\forall L \in \Ds$,} \quad w \in L \ \Rightarrow\ w' \in L \\
    w \canof w' & \text{if and only if} & \text{$\forall L \in \Fs$,} \quad w \in L \ \Rightarrow\ w' \in L
  \end{array}
\]
Clearly, these are preorder relations. In particular, given $u \in A^*$, we let $\upset[\Ds]{u} \subseteq A^*$ as the least upper set of \canod containing $u$: $\upset[\Ds]{u} = \{v \in A^* \mid u \canod v\}$. Similarly, we define $\upset[\Fs]{u}$ from \canod. It is straightforward to verify that $\upset[\Ds]{u} \in \Ds$ and $\upset[\Fs]{u} \in \Fs$ for every $u \in A^*$ (see~\cite[Section~2]{pzgenconcat} for the proofs). Finally, since \Ds and \Fs are closed under quotients, one may verify that \canod and \canof are compatible with word concatenation (again, we refer the reader to~\cite{pzgenconcat} for details).

Finally, we have the following property which holds because $\Fs \subseteq \pol{\cocl{\Ds}}$. This is a corollary of a generic property of polynomial closure presented in~\cite{pzgenconcat}.

\begin{lemma} \label{lem:pbp:choiceofk}
  There exist natural numbers $p,k \in \nat$ such that for every $\ell \geq k$ and $u,v \in A^*$ satisfying $v \leq_\Ds u$, we have,
  \[
    u^{p\ell+1} \leq_{\Fs} u^{p\ell} v u^{p\ell}
  \]
\end{lemma}

\begin{proof}
  Since $\Fs \subseteq \pol{\cocl{\Ds}}$, it follows from~\cite[Proposition~34]{pzgenconcat} that for every $L \in \Fs$, there exists $k_L,p_L \in \nat$ such that for every $\ell \geq k_L$ and every $u,v \in A^*$ satisfying $u \leqslant_{\cocl{\Ds}} v$, we have:
  \[
    u^{p_H\ell+1} \in L \quad \Rightarrow \quad u^{p_H\ell} v u^{p_H\ell} \in L.
  \]
  We choose $k$ as the maximum of all numbers $k_L$ for $L \in \Fs$ and $p$ as the least common multiplier of all numbers $p_L$ for $L \in \Fs$ (recall that \Fs is finite by definition). It then follows, that for every $\ell \geq k$ and $u,v \in A^*$ satisfying $u \leqslant_{\cocl{\Ds}} v$, we have:
  \[
    u^{p\ell+1} \in L \quad \Rightarrow \quad u^{p\ell} v u^{p\ell} \in L \quad \text{for every $L \in \Fs$.}
  \]
  By definition, this exactly says that $u^{p\ell+1} \leq_{\Fs} u^{p\ell} v u^{p\ell}$.

  Finally, observe that for every $u,v \in A^*$, the hypothesis $v \canod u$ implies that $u \leqslant_{\cocl{\Ds}} v$. Indeed, $v \canod u$ means that for every $L \in \Ds$, we have $v \in L \Rightarrow u \in L$. The contrapositive then states that for every $L \in \Ds$, we have $u \not\in L \Rightarrow v \not\in L$. Finally, since the languages of $\cocl{\Ds}$ are the complements of those in \Ds, it follows that for every $L \in \cocl{\Ds}$, we have $u \in L \Rightarrow v \in L$, i.e. $u \leqslant_{\cocl{\Ds}} v$. This concludes the proof.
\end{proof}

We may now come back to the main argument. Recall that we have $(r,T) \in \ioptic{\pbwmrats}$ and an idempotent $(e,f) \in \dclosr T \subseteq M \times R$. Furthermore, we have a \pbpol{\Cs}-cover \Kb of $L \cap \alpha\inv(e)$. Our objective is to show that $frf+f \in \prin{\rho}{\Kb}$. By definition, this amounts to exhibiting $K \in \Kb$ such that $frf+f \leq \rho(K)$. We start with the following lemma which exhibits a word $v \in A^*$ which we shall use for the construction of $K$. This is where we use the hypothesis that $(r,T) \in \ioptic{\pbwmrats}$.

\begin{lemma} \label{lem:pbp:thev}
  There exists $v \in A^*$ such that $r = \rho(v)$, $v \in L$ and $f \in \opti{\pbpol{\Cs}}{\left(\upset[\Ds]{v}\right) \cap \alpha\inv(e),\rho}$
\end{lemma}

\begin{proof}
  Since $(r,T) \in \ioptic{\pbwmrats}$, it follows from Lemma~\ref{lem:pbp:pbpolsound} that,
  \[
    T \in \opti{\pol{\Cs}}{L \cap \rho_*\inv(r),\tau}
  \]
  Since $\Ds \subseteq \pol{\Cs}$ by definition in Fact~\ref{fct:pbp:strat}, it then follows from Fact~\ref{fct:linclus} that,
  \[
    T \in \opti{\Ds}{L \cap \rho_*\inv(r),\tau}
  \]
  Let $\Hb = \{\upset[\Ds]{v} \mid v \in L \cap \rho_*\inv(r)\}$. By definition, \Hb is a finite set of languages in \Ds. Moreover \Hb is clearly a \Ds-cover of $L \cap \rho_*\inv(r)$. Consequently, we have $\opti{\Ds}{L \cap \rho_*\inv(r),\tau} \subseteq \prin{\tau}{\Hb}$. We get $H \in \Hb$ such that $T \subseteq \tau(H)$. By definition $H = \upset[\Ds]{v}$ for some $v \in A^*$ such that $\rho(v) = r$ and $v \in L$. We now unravel the definition of $\tau = \lratauxppc$ which yields,
  \[
    T \subseteq  \{(s,q) \in M \times R \mid q \in \opti{\pbpol{\Cs}}{\left(\upset[\Ds]{v}\right) \cap \alpha\inv(s),\rho}\}
  \]
  Since $(e,f) \in \dclosr T$ by hypothesis, we obtain as desired that $f \in \opti{\pbpol{\Cs}}{\left(\upset[\Ds]{v}\right) \cap \alpha\inv(e),\rho}$. This concludes the proof.
\end{proof}

We are now ready to construct the desired language $K \in \Kb$ such that $frf+f \leq \rho(K)$. We fix $v \in A^*$ as the word described in Lemma~\ref{lem:pbp:thev}. Consider the following set \Hb,
\[
  \Hb = \{\upset[\Fs]{u} \mid u  \in \left(\upset[\Ds]{v}\right) \cap \alpha\inv(e)\}
\]
Since $\Fs \subseteq \pbpol{\Cs}$ it is immediate by definition that \Hb is a finite set of languages in \pbpol{\Cs}. Therefore, \Hb is clearly a \pbpol{\Cs}-cover of $\left(\upset[\Ds]{v}\right) \cap \alpha\inv(e)$. Consequently, since $f \in \opti{\pbpol{\Cs}}{\left(\upset[\Ds]{v}\right) \cap \alpha\inv(e),\rho}$ by Lemma~\ref{lem:pbp:thev}, we have $f \in \prin{\rho}{\Hb}$. This yields a language $H \in \Hb$ such that $f \leq \rho(H)$. Consider the natural numbers $k,p \in \nat$ given by Lemma~\ref{lem:pbp:choiceofk}. We define,
\[
  G = H^{pk}\cdot \{v\} \cdot H^{pk} \cup H^{pk+1}
\]
In the following lemma, we exhibit a language $K \in \Kb$ which contains $G$. This is the desired language: using the definition of $G$ we show that $frf+f \leq \rho(K)$.

\begin{lemma} \label{lem:pbp:therightlang}
  There exists $K \in \Kb$ such that $G \subseteq K$.
\end{lemma}

Let us first use Lemma~\ref{lem:pbp:therightlang} to finish the main argument. Let $K \in \Kb$ be the language given by the lemma. Since $G \subseteq K$, we have $\rho(G) \subseteq \rho(K)$ by definition of \ratms. Moreover, since $f \leq \rho(H)$ and $f$ is idempotent, we get from the definition of $G$ that,
\[
  f \cdot \rho(v) \cdot f + f \leq \rho(G)
\]
By definition in Lemma~\ref{lem:pbp:thev}, we have $r = \rho(v)$. Therefore, we get that $frf +f \leq \rho(G) \leq \rho(K)$ which concludes the proof: we have $frf+f \in \prin{\rho}{\Kb}$. It remains to prove the lemma.

\begin{proof}[Proof of Lemma~\ref{lem:pbp:therightlang}]
  By definition of \Hb, $H = \upset[\Fs]{u}$ for some $u \in \left(\upset[\Ds]{v}\right) \cap \alpha\inv(e)$. This implies that $u \in \alpha\inv(e)$ and since $e$ is idempotent, it follows that $u^{pk+1} \in \alpha\inv(e)$. Moreover, we have $v \in L$ by definition in Lemma~\ref{lem:pbp:thev} and since $L \in \Ds$, $u \in \upset[\Ds]{v}$ yields $u \in L$. Thus, since $LL = L$, we get $u^{pk+1} \in L$. Altogether, this yields $u^{pk+1} \in L \cap \alpha\inv(e)$ and since \Kb is by definition a \pbpol{\Cs}-cover of $L \cap \alpha\inv(e)$, there exists $K \in \Kb$ such that $u^{pk+1} \in K$. We prove that $G \subseteq K$ which concludes the proof.

  Consider $w \in G$. We show that $u^{pk+1} \leq_{\Fs} w$. Since $K \in \Fs$ by definition of \Fs in Fact~\ref{fct:pbp:strat} and $u^{pk+1} \in K$, it will follow as desired that $w \in K$.  By definition, $G$ is the union of two languages. Thus,  there are two cases. Assume first that $w \in H^{pk+1}$ (i.e. $w$ is the concatenation of $pk+1$ words in $H$). Since $H = \upset[\Fs]{u}$ and $\leq_{\Fs}$ is compatible with concatenation, it is immediate that this implies $u^{pk+1} \leq_{\Fs} w$, finishing the proof for this case.

  We now assume that $w \in H^{pk}\cdot \{v\} \cdot H^{pk}$. It follows that $w = x v y$ with $x,y \in H^{pk}$. We may reuse the fact that $\leq_{\Fs}$ is compatible with concatenation to obtain $u^{pk} \leq_{\Fs} x$ and $u^{pk} \leq_{\Fs} y$. Consequently, we get $u^{pk}vu^{pk} \leq_{\Fs} xvy$ and by transitivity, it suffices to show that $u^{pk+1} \leq_{\Fs} u^{pk}vu^{pk}$. This is immediate from Lemma~\ref{lem:pbp:choiceofk} as we have $u \in \upset[\Ds]{v}$ which means that $v \canod u$.
\end{proof}

\subsection{Completeness}

We have proved that \ioptic{\lratauxppc} is a \pbpol{\Cs}-complete subset of $M \times R$. It remains to show that it is the least such subset. Therefore, we fix an arbitrary \pbpol{\Cs}-complete set $S \subseteq M \times R$ and show that $\ioptic{\lratauxppc} \subseteq S$.

By definition, $\ioptic{\lratauxppc} \subseteq \lratauxppc(L)$ for every language $L \in \Cs$ such that $\veps \in L$. Hence, it suffices to exhibit a language $L \in \Cs$ such that $\veps \in L$ and $\lratauxppc(L) \subseteq S$. We first choose the appropriate language $L$.

We let $H$ as a \Cs-optimal \iden for the \nice \mratm \pbwmrats. That is, we have $H \in \Cs$, $\veps \in H$ and $\pbwmrats(H) = \ioptic{\pbwmrats}$. Since $H \in \Cs$ is a group language and \Cs is a \vari, it is standard that $H$ is recognized by a morphism $\psi: A^* \to G$ into a finite group $G$ such that every language recognized by $\psi$ belongs to \Cs. (it suffices to choose $\psi$ as the \emph{syntactic morphism} of $H$, see~\cite{pingoodref} for example). In particular, the language $L = \psi\inv(1_G)$ belongs to \Cs and contains \veps by definition. It now remains to show the following inclusion:
\begin{equation} \label{eq:pbp:ginc}
  \lratauxppc(L) \subseteq S
\end{equation}
The argument reuses the main theorem of~\cite{pseps3j} which characterizes the set \popti{\pbpol{\Gs}}{\alpha}{\rho} when \Gs is a \textbf{finite} \vari. Let us first define \Gs.

\medskip

We let \Gs as the class containing all languages recognized by $\psi$ (which is included in \Cs by definition). It is immediate that \Gs is a finite \vari. Moreover, the quotient set ${A^*}/{\sim_{\Gs}}$ is isomorphic to the group $G$ and the morphism $w \mapsto \typ{w}{\Gs}$ from $A^*$ to ${A^*}/{\sim_{\Gs}}$ corresponds to the morphism $\psi: A^* \to G$. Hence, for the sake of avoiding clutter, we identify ${A^*}/{\sim_{\Gs}}$ with $G$ in the proof. We may now recall the theorem of~\cite{pseps3j}. Note that we slightly adapt the formulation for the sake of convenience. In particular, the theorem of~\cite{pseps3j} requires manipulating monoid morphisms and \mratms which are ``\Gs-compatible'' (essentially, this means that they compute $\sim_{\Gs}$-classes). Here, we avoid this requirement (which is not satisfied by $\alpha$ and $\rho$) by manipulating  $\sim_{\Gs}$-classes (which we identify with elements of $G$) explicitly.

Consider two sets $S' \subseteq G \times M \times R$ and $\Ts \subseteq G \times R \times 2^{M \times R}$. We say that the pair $(S',\Ts)$ is \pbpol{\Gs}-saturated if the following conditions are satisfied:
\begin{itemize}
\item Properties on $S' \subseteq G \times M \times R$:
  \begin{enumerate}
  \item {\bf Trivial elements:} For every $w \in A^*$, we have $(\psi(w),\alpha(w),\rho(w)) \in S'$.
  \item {\bf Downset:} We have $\dclosr S' = S'$.
  \item {\bf Multiplication:} For every $(g,s,q),(h,t,r) \in S'$, we have $(gh,st,qr) \in S'$.
  \item {\bf \pbpol{\Gs}-closure.}\label{op:pbp:clos} For every $(g,r,U) \in \Ts$ where $g \in G$ is idempotent (i.e, $g = 1_G$) and every pair of idempotents $(e,f) \in \dclosr U \subseteq M \times R$, we have $(g,e, f \cdot (r + 1_R)  \cdot f) \in S'$.
  \end{enumerate}
\item Properties on $\Ts\subseteq G \times R \times 2^{M \times R}$:
  \begin{enumerate}
    \setcounter{enumi}{4}
  \item {\bf Trivial elements:} For every $w \in A^*$, we have $(\psi(w),\rho(w),\{(\alpha(w),\rho(w))\}) \in \Ts$.
  \item {\bf Multiplication:} For every $(g,q,U),(h,r,V) \in \Ts$, we have $(gh,qr,UV) \in \Ts$.
  \item {\bf Nested \pol{\Cs}-closure.}\label{op:pbp:nest} For every triple of idempotents $(g,f,F) \in \Ts$, we have $(g,f,F \cdot S'(g) \cdot F) \in \Ts$ (where $S'(g) = \{(s,r) \mid M \times R \mid (g,s,r) \in S'\}$).
  \end{enumerate}
\end{itemize}

The following proposition can be verified from~\cite[Theorem~7.11]{pseps3j} (it is a reformulation of this result with our terminology).

\begin{proposition} \label{prop:pbp:algo}
  Consider the least \pbpol{\Gs}-saturated pair $(S',\Ts)$ (with respect to pairwise inclusion). Then, for every $g \in G$ and every $s \in M$, the following equality holds:
  \[
    \opti{\pbpol{\Gs}}{\psi\inv(g) \cap \alpha\inv(s),\rho} = \{r \in R \mid (g,s,r) \in S'\}
  \]
\end{proposition}

We now use Proposition~\ref{prop:pbp:algo} to prove that~\eqref{eq:pbp:ginc} holds (i.e. $\lratauxppc(L) \subseteq S$). First, we exhibit a specific  \pbpol{\Gs}-saturated pair $(S',\Ts)$ which we build from $S$ and the \nice \mratm \pbwmrats. Then, we apply the proposition to this pair and use the information we get to deduce~\eqref{eq:pbp:ginc}.

\medskip
\noindent
{\bf Definition of $(S',\Ts)$.}  We start with the definition of the set $\Ts \subseteq G \times R \times 2^{M \times R}$:
\[
  \Ts = \{(1_G,1_R,S)\} \cup \{(g,q,T) \mid (q,T) \in \pbwmrats(\psi\inv(g))\}
\]
We may now define $S' \subseteq G \times M \times R$. We let,
\[
  S' =  \bigcup_{(g,q,T) \in \Ts} \{(g,s,r) \mid (s,r) \in \dclosr T\}
\]
Before we prove that the pair $(S',\Ts)$ is \pbpol{\Gs}-saturated, let us present two simple facts about $S'$ and \Ts that will be useful.

\begin{fact} \label{fct:pbp:gfact1}
  For every $(g,q,T) \in \Ts$, one of the two following properties hold:
  \begin{itemize}
  \item $(g,q,T) = (1_G,1_R,\{(1_M,1_R)\})$, or,
  \item we have $STS = ST = TS = T$.
  \end{itemize}
\end{fact}

\begin{proof}
  Let us make two preliminary observations. First, we have $(1_M,1_R) \in S$. Indeed, we know that $S$ is \pbpol{\Cs}-complete which implies that for every $(r,T) \in \ioptic{\pbwmrats} = \pbwmrats(H)$, we have $T \subseteq S$. Since $\veps \in H$ and $\pbwmrat{\veps} = \{(1_R,\{(1_M,1_R)\})\}$, it follows that $(1_M,1_R) \in S$. Moreover, we know that $SS = S$. Indeed, since $S$ is \pbpol{\Cs}-complete, we have $S \cdot SS$ and the converse inclusion is immediate since $(1_M,1_R) \in S$.

  We may now prove the fact. Consider $(g,q,T) \in \Ts$. By definition of \Ts, there are two cases. First, if $(g,q,T) = (1_G,1_R,S)$, it is immediate that second assertion in the fact holds since $SS = S$. Otherwise, we have $(q,T) \in \pbwmrats(\psi\inv(g))$. By definition of \pbwmrats, it is immediate that there exists some word $w \in \psi\inv(g)$ such that $\pbwmrats(w) = \{(q,T)\}$. If $w = \veps$, then $g = 1_G$, $q= 1_R$ and $T = \{(1_M,1_R)\}$: the first assertion in the fact holds. Otherwise $w \in A^+$ and since $SS = S$ it is immediate from the definition of \pbwmrats that $STS = ST = TS = T$: the second assertion in the fact holds.
\end{proof}

\begin{fact} \label{fct:pbp:gfact2}
  We have $S'(1_G) = S$.
\end{fact}

\begin{proof}
  We first show that $S \subseteq S'(1_G)$. By definition, $(1_G,1_R,S) \in \Ts$ which implies that $S \subseteq S'(1_G)$ by definition of $S'$. it remains to prove that $S'(1_G) \subseteq S$. let $(s,r) \in S'(1_G)$. By definition, we have $(q,T) \in R \times 2^{M \times R}$ such that $(1_G,q,T) \in \Ts$ and $(s,r) \in \dclosr T$.  By definition of \Ts, there are now two cases. First, if $(1_G,q,T) = (1_G,1_R,S)$, it is immediate that $(s,r) \in \dclosr S$ which concludes the proof: since $S$ is \pbpol{\Cs}-saturated, we have $\dclosr S = S$. Otherwise, we have $(q,T) \in \pbwmrats(\psi\inv(1_G))$. Recall that $\psi$ recognizes $H$ and $\veps \in H$. Thus, $\psi\inv(1_G) \subseteq H$ which yields $(q,T) \in \pbwmrats(H)$. Since $H$ is a \Cs-optimal \iden for \pbwmrats by definition, it follows that $(q,T) \in \ioptic{\pbwmrats}$. Thus, since $S$ is \pbpol{\Cs}-complete by definition, we have $\dclosr T \subseteq \dclosr S = S$ and we conclude that $(s,r) \in S$, finishing the proof.
\end{proof}

\medskip
\noindent
{\bf Main argument.} We may now prove the inclusion described in~\eqref{eq:pbp:ginc}. Using the hypothesis that $S \subseteq M \times R$ is \pbpol{\Cs}-complete, one may prove the following proposition.

\begin{proposition} \label{prop:pbp:compgmain}
  The pair $(S',\Ts)$ is \pbpol{\Gs}-saturated.
\end{proposition}

Let us first use Proposition~\ref{prop:pbp:compgmain} to show that~\eqref{eq:pbp:ginc} holds and conclude the main argument. We have to show that $\lratauxppc(L) \subseteq S$. Consider $(s,r) \in \lratauxppc(L)$. We show that $(s,r) \in S$. By definition, we have,
\[
  r \in \opti{\pbpol{\Cs}}{\alpha\inv(s) \cap L,\rho}
\]
Since $\Gs \subseteq \Cs$, it is immediate that $\pbpol{\Gs} \subseteq \pbpol{\Cs}$. Therefore, it is clear that,
\[
  r \in \opti{\pbpol{\Gs}}{\alpha\inv(s) \cap L,\rho}
\]
Moreover, we have $L = \psi\inv(1_G)$ by definition. Therefore, since $(S',\Ts)$ is \pbpol{\Gs}-saturated, it is immediate from Proposition~\ref{prop:pbp:algo} that,
\[
  (1_G,s,r) \in S'
\]
By Fact~\ref{fct:pbp:gfact2}, this yield $(s,r) \in S$, finishing the completeness proof. It remains to present an argument for Proposition~\ref{prop:pbp:compgmain}.

\begin{proof}[Proof of Proposition~\ref{prop:pbp:compgmain}]
  We prove that $(S',\Ts)$ is \pbpol{\Gs}-saturated. There are seven properties to verify: four apply to $S'$ and three to \Ts. We start with the latter: we must show that \Ts contains the trivial elements and is closed under multiplication and nested \pol{\Gs}-closure.

  \medskip
  \noindent
  {\bf Multiplication for \Ts.} Consider $(g,q,T),(g',q',T') \in \Ts$, we have to show that $(gg',qq',TT') \in \Ts$. Since \Ts is defined as the union of two sets, there are several cases. Assume first that either $(g,q,T)$ or $(g',q',T')$ is equal to $(1_G,1_R,S)$. By symmetry, we consider the later case: $(g',q',T') = (1_G,1_R,S)$. it follows that $(gg',qq',TT') = (g,q,TS)$. In view of Fact~\ref{fct:pbp:gfact1}, there are now two cases:
  \begin{itemize}
  \item First, if we have $(g,q,T) = (1_G,1_R,\{(1_M,1_R)\})$, then $(gg',qq',TT') = (1_G,1_R,S) \in \Ts$.
  \item Otherwise, $TS = T$ and  $(gg',qq',TT') = (g,q,T) \in \Ts$.
  \end{itemize}
  Now assume that $(q,T) \in \pbwmrats(\psi\inv(g))$ and $(q',T') \in \pbwmrats(\psi\inv(g'))$. Under these hypotheses, we obtain that, $(qq',TT') \in \pbwmrats(\psi\inv(g)\psi\inv(g'))$. Since $\psi$ is a morphism, we have  $\psi\inv(g)\psi\inv(g') \subseteq \psi\inv(gg')$ which yields $(qq',TT') \in \pbwmrats(\psi\inv(gg'))$. By definition, this implies $(gg',qq',TT') \in \Ts$.

  \medskip
  \noindent
  {\bf Trivial elements for \Ts.} Consider $w \in A^*$, we show that $(\psi(w),\rho(w),\{\alpha(w),\rho(w)\}) \in \Ts$. Assume first that $w = \veps$. In that case, observe that $\pbwmrats(\veps) = \{(1_R,\{(1_M,1_R)\})\}$ (this is the neutral element of $2^{R \times 2^{M \times R}}$). Hence, since $\veps \in \psi\inv(\psi(\veps))$, we have $(\rho(\veps),\{(\alpha(\veps),\rho(\veps)\}) \in \pbwmrats(\psi\inv(\psi(\veps)))$ which yields that $(\psi(\veps),\rho(\veps),\{\alpha(\veps),\rho(\veps)\}) \in \Ts$.

  Assume now that $w \in A^+$. Since we already proved that \Ts is closed under multiplication, it suffices to consider the case when $w = a \in A$. We already observed that $(1_M,1_R) \in S$ (see the proof of Fact~\ref{fct:pbp:gfact1}). Therefore, it is immediate that $(\rho(a),\{(\alpha(a),\rho(a))\}) \in \pbwmrats(a)$ by definition. Hence, since $a \in  \psi\inv(\psi(a))$, we have $(\rho(a),\{(\alpha(a),\rho(a)\}) \in \pbwmrats(\psi\inv(\psi(a)))$ which yields that $(\psi(a),\rho(a),\{\alpha(a),\rho(a)\}) \in \Ts$.

  \medskip
  \noindent
  {\bf Nested \pol{\Gs}-closure.}  Consider a triple of idempotents $(1_G,f,F) \in \Ts$ (note that the first element is $1_G$ since the only idempotent in a group is its neutral element). We have to show that,
  \[
    (1_G,f,F \cdot S'(1_G) \cdot F) \in \Ts
  \]
  By Fact~\ref{fct:pbp:gfact2}, we know that $S'(1_G) = S$. Hence, we have to show that $(1_G,f,F \cdot S \cdot F) \in \Ts$. In view of Fact~\ref{fct:pbp:gfact1}, there are two cases:
  \begin{itemize}
  \item First, if we have $(1_G,f,F) = (1_G,1_R,\{(1_M,1_R)\})$, then $(1_G,f,F \cdot S \cdot F) = (1_G,1_R,S) \in \Ts$.
  \item Otherwise, $FSF = F$ and  $(1_G,f,F \cdot S \cdot F) = (1_G,f,F) \in \Ts$.
  \end{itemize}

  \medskip

  We turn to the properties concerning $S'$: we have to show that $S'$ contains the trivial elements and is closed under downset, multiplication and \pbpol{\Gs}-closure.

  \medskip
  \noindent
  {\bf Trivial elements for $S'$.} Consider $w \in A^*$, we have to show that $(\psi(w),\alpha(w),\rho(w)) \in S'$. We already show that $(\psi(w),\rho(w),\{\alpha(w),\rho(w)\}) \in \Ts$. Thus, it is immediate by definition of $S'$ from \Ts that $(\psi(w),\alpha(w),\rho(w)) \in S'$.

  \medskip
  \noindent
  {\bf Downset for $S'$.} It is immediate from the definition of $S'$ that $S' = \dclosr S'$ (the closure under downset is built-in).

  \medskip
  \noindent
  {\bf Multiplication for $S'$.} Consider $(g_1,s_1,r_1),(g_2,s_2,r_2) \in S'$, we show that $(g_1g_2,s_1s_2,r_1r_2) \in S'$. By definition of $S'$, we have $(q_1,T_1),(q_2,T_2) \in R \times 2^{M \times R}$ such that $(g_i,q_i,T_i) \in \Ts$ and $(s_i,r_i) \in \dclosr T_i$ for $i = 1,2$. Since we already established that \Ts is closed under multiplication, we have $(g_1q_2,q_1q_2,T_1T_2) \in \Ts$. Moreover, it is clear that $(s_1q_2,r_1r_2) \in (\dclosr T_1)(\dclosr T_2) \subseteq \dclosr (T_1T_2)$. Hence, we get as desired that $(g_1g_2,s_1s_2,r_1r_2) \in S'$ by definition of $S'$.

  \medskip
  \noindent
  {\bf \pbpol{\Gs}-closure.} Consider an element $(1_G,r,U) \in \Ts$ and a pair of idempotents $(e,f) \in \dclosr U$, we have to show that,
  \[
    (1_G,e,f \cdot (1_R + r) \cdot f) \in S'
  \]
  By definition of \Ts, there are two cases. First, it $(1_G,r,U) = (1_G,1_R,S)$, then $(e,f) \in \dclosr S$ and $(1_G,e,f \cdot (1_R + r) \cdot f) = (1_G,e,f)$. It is now immediate by definition of $S'$ from \Ts that $(1_G,e,f) \in S'$.

  Otherwise, we have $(r,U) \in \pbwmrats(\psi\inv(1_G))$.  Recall that $\psi$ recognizes $H$ and $\veps \in H$. Thus, $\psi\inv(1_G) \subseteq H$ which yields $(r,U) \in \pbwmrats(H)$. Since $H$ is a \Cs-optimal \iden for \pbwmrats by definition, it follows that $(r,U) \in \ioptic{\pbwmrats}$. Thus, since $S$ is \pbpol{\Cs}-complete by definition, $(e,f) \in \dclosr U$ implies that $(e,f \cdot (1_R + r) \cdot f) \in S$. Since $ (1_G,1_R,S) \in \Ts$, this implies that $(1_G,e,f \cdot (1_R + r) \cdot f) \in S'$ by definition of $S'$ from \Ts, finishing the proof.
\end{proof}

\end{document}